\documentclass{article}

\usepackage{microtype}
\usepackage{graphicx}
\usepackage{subcaption}
\usepackage{booktabs} 

\usepackage{hyperref}

\usepackage[preprint]{icml2026}

\usepackage{amsmath}
\usepackage{amssymb}
\usepackage{mathtools}
\usepackage{amsthm}
\usepackage{amsfonts}
\usepackage{multirow}
\usepackage{graphicx}
\usepackage{booktabs}
\usepackage{subcaption}

\usepackage[capitalize,noabbrev]{cleveref}

\theoremstyle{plain}
\newtheorem{theorem}{Theorem}[section]

\theoremstyle{definition}

\theoremstyle{remark}

\usepackage[textsize=tiny]{todonotes}

\icmltitlerunning{HPE: Hallucinated Positive Entanglement for Backdoor Attacks in Federated Self-Supervised Learning}

\begin{document}

\twocolumn[
  \icmltitle{HPE: Hallucinated Positive Entanglement for Backdoor Attacks in Federated Self-Supervised Learning}



  \icmlsetsymbol{equal}{*}

  \begin{icmlauthorlist}
    \icmlauthor{Jiayao Wang}{1}
    \icmlauthor{Yang Song}{1}
    \icmlauthor{Zhendong Zhao}{2}
    \icmlauthor{Jiale Zhang}{1}
    \icmlauthor{Qilin Wu}{3}
    \icmlauthor{Wenliang Yuan}{4}
    \icmlauthor{Junwu Zhu}{1}
    \icmlauthor{Dongfang Zhao}{5}
  \end{icmlauthorlist}

  \icmlaffiliation{1}{School of Information Engineering, Yangzhou University, China}
  \icmlaffiliation{2}{Institute of Information Engineering, Chinese Academy of Sciences, China}
  \icmlaffiliation{3}{School of Computing and Artificial Intelligence, Chaohu University, China}
  \icmlaffiliation{4}{College of Data Science, Jiaxing University, China}
  \icmlaffiliation{5}{Tacoma School of Engineering and Technology, University of Washington, USA}

  \icmlcorrespondingauthor{Junwu Zhu}{jwzhu@yzu.edu.cn}

  \icmlkeywords{Machine Learning, ICML}

  \vskip 0.3in
]



\printAffiliationsAndNotice{}  

\begin{abstract}
Federated self-supervised learning (FSSL) enables collaborative training of self-supervised representation models without sharing raw unlabeled data. 
While it serves as a crucial paradigm for privacy-preserving learning, its security remains vulnerable to backdoor attacks, where malicious clients manipulate local training to inject targeted backdoors. 
Existing FSSL attack methods, however, often suffer from low utilization of poisoned samples, limited transferability, and weak persistence.
To address these limitations, we propose a new backdoor attack method for FSSL, namely Hallucinated Positive Entanglement (HPE). 
HPE first employs hallucination-based augmentation using synthetic positive samples to enhance the encoder’s embedding of backdoor features. 
It then introduces feature entanglement to enforce tight binding between triggers and backdoor samples in the representation space. 
Finally, selective parameter poisoning and proximity-aware updates constrain the poisoned model within the vicinity of the global model, enhancing its stability and persistence.
Experimental results on several FSSL scenarios and datasets show that HPE significantly outperforms existing backdoor attack methods in performance and exhibits strong robustness under various defense mechanisms.
\end{abstract}

\section{Introduction}
Self-supervised learning (SSL)~\cite{SSL01Moco,SSL02Simclr,SSL03BYOL} offers a scalable solution to representation learning by leveraging supervisory signals inherently available in unlabeled data, thus eliminating the need for manual annotations.
Its core idea is to design pretext tasks or contrastive objectives to guide the model in learning discriminative and generalizable representations. 
SSL typically employs a Siamese network structure to minimize the feature distance between different views of the same instance, thereby enhancing semantic consistency. 
This paradigm has demonstrated performance comparable to supervised learning across various domains, including vision~\cite{SSL01Moco}, language~\cite{language}, and speech~\cite{speech}.

\begin{figure}[t]
  \centering
  \includegraphics[width=1\linewidth]{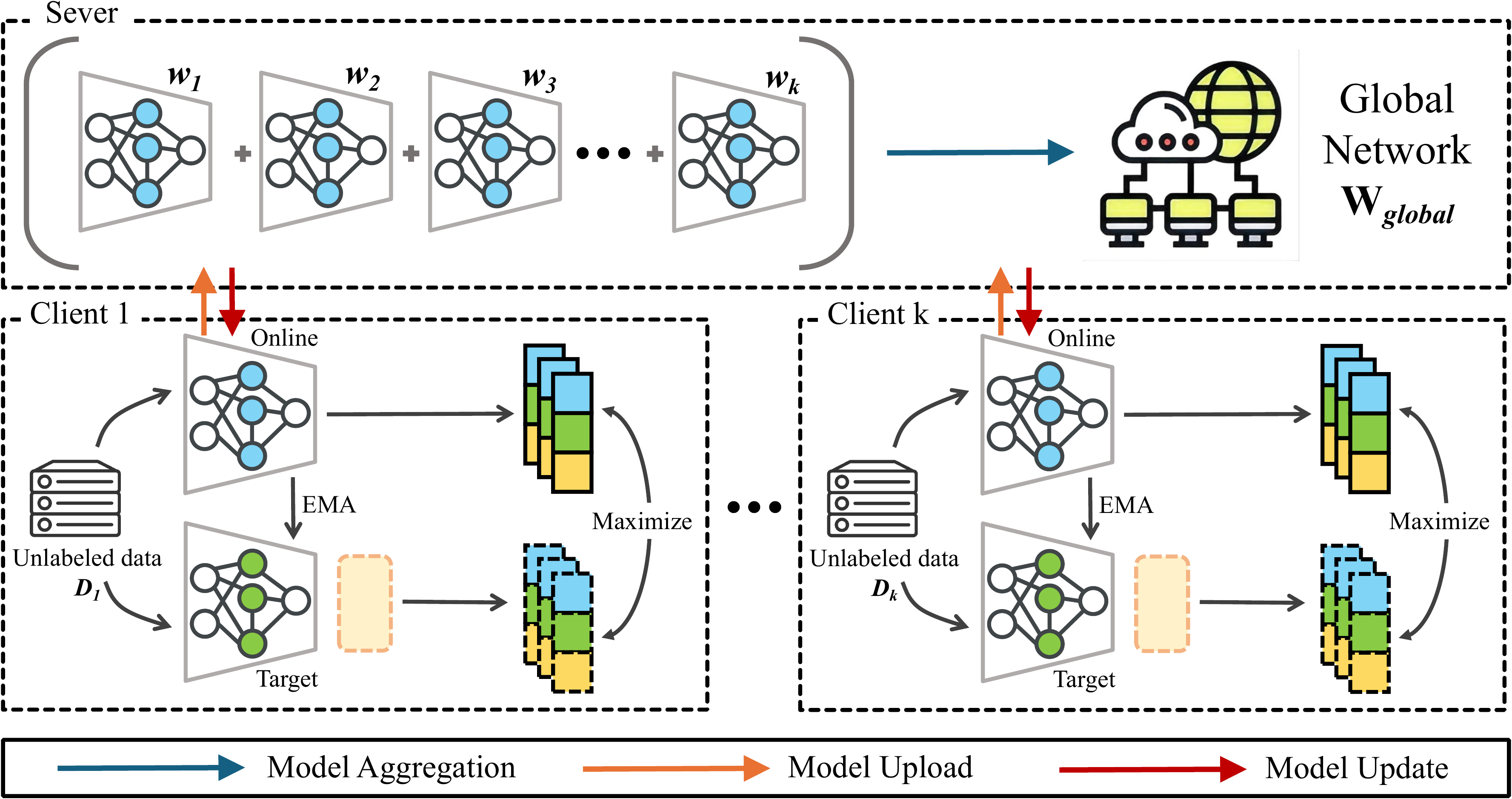}
   \caption{Framework of federated self-supervised learning.}
   \label{fig:one}
\end{figure}

Although SSL has achieved impressive results across diverse tasks, its success often hinges on access to large-scale unlabeled data.
In resource-constrained or privacy-sensitive scenarios, this dependency may become a performance bottleneck. 
Moreover, centralized data collection raises growing concerns about privacy breaches and regulatory compliance~\cite{security}. 
In this context, federated self-supervised learning (FSSL) introduces a distributed training paradigm that alleviates data isolation and privacy risks, enabling effective SSL model training without centralized data aggregation~\cite{FSSL01,FSSL02}.
\Cref{fig:one} presents an overview of the generic FSSL paradigm, which comprises an end-to-end training pipeline with the following steps: 
1) \textbf{Local Training:} Each client $i$, $1 \le i \le k$, trains locally on its unlabeled data $D_i$ using an online encoder and a target encoder.
2) \textbf{Model Upload:} The client $i$ uploads its trained online encoder to the server.
3) \textbf{Model Aggregation:} The server aggregates all client encoders to generate a global encoder $W_{global}$.
4) \textbf{Model Update:} Clients update their local online encoders with the global encoder for the next round of training.
FSSL executes several iterations of the above steps until the global model converges and achieves satisfactory performance.

However, the decentralized architecture of FSSL may inherit the security vulnerabilities of traditional federated supervised learning (FSL). 
Recent studies~\cite{bafssl01BADFSS,bafssl02EmInspector} have shown that FSSL is vulnerable to backdoor attacks, where malicious clients inject trigger-embedded poisoned samples during local training, covertly implanting backdoor features that compromise the global encoder through model upload and aggregation.
Specifically, Zhang \emph{et al.} proposed BADFSS~\cite{bafssl01BADFSS}, which implants backdoor triggers via supervised contrastive learning and attention alignment. 
Although effective, BADFSS assumes identical datasets for pre-training and downstream tasks, limiting its applicability in real-world scenarios. 
Moreover, its reliance on a high poisoning rate reduces both stealth and efficiency.
Qian \emph{et al.} introduced EmInspector~\cite{bafssl02EmInspector}, a defense framework that detects backdoors by analyzing local embedding spaces and formalizes two attack modes: single-pattern and coordinated-pattern. 
While the latter improves attack robustness via multi-client cooperation, EmInspector fails to consistently preserve backdoor features across rounds due to aggregation-induced dilution.

Based on the above analysis, backdoor attacks in FSSL face two key challenges: (i) low utilization efficiency of poisoned samples, limiting their impact under data-scarce conditions; and (ii) dilution of backdoor features during multi-round aggregation, undermining attack persistence.
To overcome these challenges, we propose HPE, a novel backdoor attack framework designed to unify effective injection with persistent propagation.
HPE first incorporates a hallucination enhancement module to generate synthetic positive samples, thereby amplifying the influence of backdoor cues within the learned representation space. 
It then introduces a feature entanglement mechanism to reinforce the semantic alignment between poisoned instances and their triggers. 
Finally, a dual-constrained poisoning strategy is employed to preserve proximity to the global model while concentrating perturbations on low-variance parameters, enhancing both the stealth and durability of the attack.

In summary, this paper makes the following contributions:
\begin{itemize}
    \item We propose HPE (Hallucinated Positive Entanglement), a backdoor attack framework tailored for visual tasks in FSSL. HPE integrates a hallucination enhancement module to synthesize positive samples and employs a feature entanglement mechanism to effectively improve the utilization of limited poisoned data and the consistency of backdoor triggers.

    \item To mitigate the dilution of backdoor signals during aggregation, we propose a dual-constrained poisoning strategy that constrains parameter divergence from the global model and injects perturbations into low-variance dimensions. This design facilitates stealthy implantation and reliable propagation of the backdoor, enhancing both its robustness and persistence across training rounds.

    \item We conduct extensive experiments on five widely-used vision benchmark datasets: CIFAR-10, GTSRB, STL-10, CIFAR-100 and ImageNet-100. Results demonstrate that HPE consistently outperforms state-of-the-art methods, demonstrating strong stability and generalizability. Moreover, HPE effectively evades existing aggregation-based defenses and backdoor detection mechanisms.
\end{itemize}

\section{Related Work}

\textbf{Federated Self-Supervised Learning.} 
An FSSL system typically comprises two key components: a client-side pretext task that extracts representations from unlabeled data, and a server-side aggregation algorithm that integrates local models uploaded by clients. 
These components jointly enable effective distributed learning in the absence of labeled data.
Among pretext tasks, contrastive learning (CL) is widely adopted, as exemplified by methods like SimCLR~\cite{SSL02Simclr} and MoCo~\cite{SSL01Moco}.
CL trains the model to distinguish between positive and negative pairs, enabling it to learn meaningful feature representations and capture the underlying structure of the data.
For aggregation algorithms, FedAVG~\cite{FL01} remains the dominant approach, averaging client models into a global update. 
Given the Non-IID nature of client data, consistency in model aggregation has attracted considerable interest~\cite{FSSL01,FSSL02,FedMKD}. Nevertheless, prior work primarily addresses data heterogeneity, with limited attention to backdoor threats in FSSL.\\
\textbf{Backdoor Attacks in FL \& SSL.} 
In FL backdoor attacks, the attacker aims to manipulate local training data or model updates of selected clients, causing the global model to behave maliciously when presented with specific trigger inputs. 
Extensive studies~\cite{BAFL01,BAFL02,BAFL03} have shown that FL is highly susceptible to backdoor attacks.
Meanwhile, in SSL, the absence of labels renders traditional label-based backdoor strategies inapplicable, prompting the development of unsupervised poisoning methods under centralized settings~\cite{BASSL,Badencoder,CTRL}, which aim to implant backdoors in the latent representation space.
However, existing methods mainly target FSL or centralized SSL scenarios, leaving backdoor attacks in FSSL largely underexplored.\\
\textbf{Backdoor Defenses.} 
Backdoor attacks have drawn increasing attention due to their stealth and high impact, posing serious challenges to model robustness during training.
Existing defenses can be broadly categorized into three types.
Model-level defenses, such as Neural Cleanse~\cite{Neural_cleanse} and DECREE~\cite{DECREE}, aim to reverse-engineer triggers and recover target labels, where successful reconstruction indicates a compromised encoder. 
Sample-level defenses, like Beatrix~\cite{Beatrix} and GradCAM~\cite{gradcam}, detect poisoned inputs by evaluating their impact on model behavior, through prediction deviation and activation localization, respectively.
Aggregation-level defenses analyze model updates in FL settings, such as FLARE~\cite{flare}, Foolsgold~\cite{Foolsgold}, and FLAME~\cite{FLAME}.
We employ the above state-of-the-art defenses to evaluate our proposed attack.

\section{Preliminary}

\subsection{Formal Problem Definition}
\textbf{FSSL System Model.}
We consider a typical FSSL setting where \(K\) clients collaboratively train a global SSL model without sharing raw data.
In each communication round, client \(i\) trains a local SSL model using its unlabeled dataset $D_i = \{x_1, \ldots , x_{|D_i|}\}, i \in \{1, 2, \ldots , K\}$.
In this work, we adopted the MoCo framework, a contrastive learning method known for efficient and effective representation learning. 
The core idea is to learn robust feature representations by maximizing the agreement between two augmented views (positives) of the same image while distinguishing them from other images (negatives).
Given input \(x\), MoCo applies two random augmentations $A_q, A_k \in \mathcal{A}$ to generate a query $x_q$ and a key $x_k$. 
These are encoded by a query encoder \(f\) and a momentum-updated key encoder \(g\) to obtain features \( v_{q} = f(x_q),\) and \(v_{k} = g(x_k)\).
MoCo also maintains a queue $Q = \{v_i^Q\}$ of \(M\) past key features as negatives.
The model is trained using the InfoNCE loss:
\begin{equation}
\mathcal{L}_{CL} = - \log \frac{\exp(v_q^\top \cdot v_k^{+} / \tau)}{\exp(v_q^\top \cdot v_k^{+} / \tau) + \sum_{i=1}^M \exp(v_q^\top  \cdot v_i^Q / \tau)},
\end{equation}
where $v_k^+$ is the positive corresponding to $v_q$, $\tau$ is the temperature parameter and exp(\(\cdot\)) is the natural exponential function. 
This contrastive objective enables the model to learn discriminative representations from unlabeled data in a federated setting.\\
\textbf{Backdoor Attack Objective.}
In FSSL, the goal of backdoor attacks is to establish an implicit association between a trigger pattern \( \Delta \) and a target class \( y_t \), such that during downstream tasks, inputs containing the trigger are misclassified as \( y_t \). This attack remains effective even under unsupervised training conditions.

To achieve this, the attacker selects several samples from the target class \( y_t \) within the local dataset, embeds the trigger \( \Delta \), and constructs poisoned samples \( x \oplus \Delta \).
During training, the attacker optimizes the parameters \( \theta \) of the local encoder \(E\)  such that the representation of the poisoned sample \(E\left(x \oplus \Delta ; \theta\right)\) is as close as possible to the representation of clean samples from the target class \(E\left(x_{y_t} ; \theta\right)\).
This objective can be formulated as:
\begin{align}
\arg\max\limits_{\theta} \frac{1}{\left|D_{p}\right|} \sum_{x \in D_{p}} s\left(E\left(x \oplus \Delta ; \theta\right), E\left(x_{y_t} ; \theta\right)\right),
\end{align}\\
where $s(\cdot,\cdot)$ denotes a similarity function, and $E$ is the encoder.
Note that we only maximize the feature similarity between backdoor samples and the target class, rather than directly associating samples with a specific label, as we have no knowledge of the downstream tasks. 

To preserve the stealthiness of the attack, the attacker must also ensure that the representations of clean inputs produced by the backdoored model remain consistent with those of the benign model. This objective can be formalized as follows:
\begin{align}
\arg \max_{\theta} \frac{1}{|D_c|} \sum_{x \in D_c} s(E_b(x), E_c(x)),
\end{align}\\
where $D_c$ is the set of clean samples, and $E_b$, $E_c$ denote the backdoored and clean encoders, respectively.
\subsection{Threat Model}
Following prior work~\cite{bafssl01BADFSS,bafssl02EmInspector}, we characterize our threat model in terms of the attacker's goals and capabilities.\\
\textbf{Attacker Objective.}
The attacker’s goal is to embed a backdoor into the encoder during pre-training, causing downstream classifiers to misclassify trigger inputs as the target class. 
In SSL, this requires implicitly associating the trigger with the target class due to the absence of labels.
In FSSL, the attacker further aims to make the global encoder inherit the backdoor, with two key objectives: \textit{stealthiness goal}, ensuring the backdoor evades detection and preserves clean accuracy; and \textit{persistence goal}, maintaining effectiveness across training rounds even if the encoder is updated.\\
\textbf{Attacker Capability.} 
The adversary masquerades as a benign participant with knowledge of the FL system, including the model architecture, global model, and local datasets. Furthermore, the adversary can manipulate local datasets (e.g., embed triggers), redesign the loss function, and arbitrarily modify the model training process. However, the attacker typically lacks access to other critical components of the FSSL system, such as the models of benign clients and the server-side aggregation rules.

\section{Methodology}
\subsection{Design of HPE}
\Cref{fig:two} illustrates the HPE framework, which centers on two key components: hallucination enhancement and backdoor feature entanglement.\\
\textbf{Motivation.}
Backdoor injection in FSSL presents unique challenges distinct from those in centralized SSL and FSL.
Attackers are constrained by limited poisoning sources and lack label supervision to anchor trigger semantics effectively.
Moreover, federated aggregation dilutes backdoor signals, while random augmentations (e.g., cropping) disrupt trigger consistency in latent representations.\\
\textbf{Hallucination Enhancement.} 
To amplify the poisoning effect under limited access to poisoned data, we propose a hallucination enhancement module, inspired by~\cite{hallucination,halp}, which synthesizes hard positive samples to enrich latent backdoor representations.
The core idea is to expand the distribution of poisoned samples in the feature space while preserving semantic consistency, thus strengthening the encoder’s ability to embed latent backdoor patterns, with detailed implementation procedures provided in \Cref{alg:algorithm}.

\begin{figure}[t]
  \centering
  \includegraphics[width=1\linewidth]{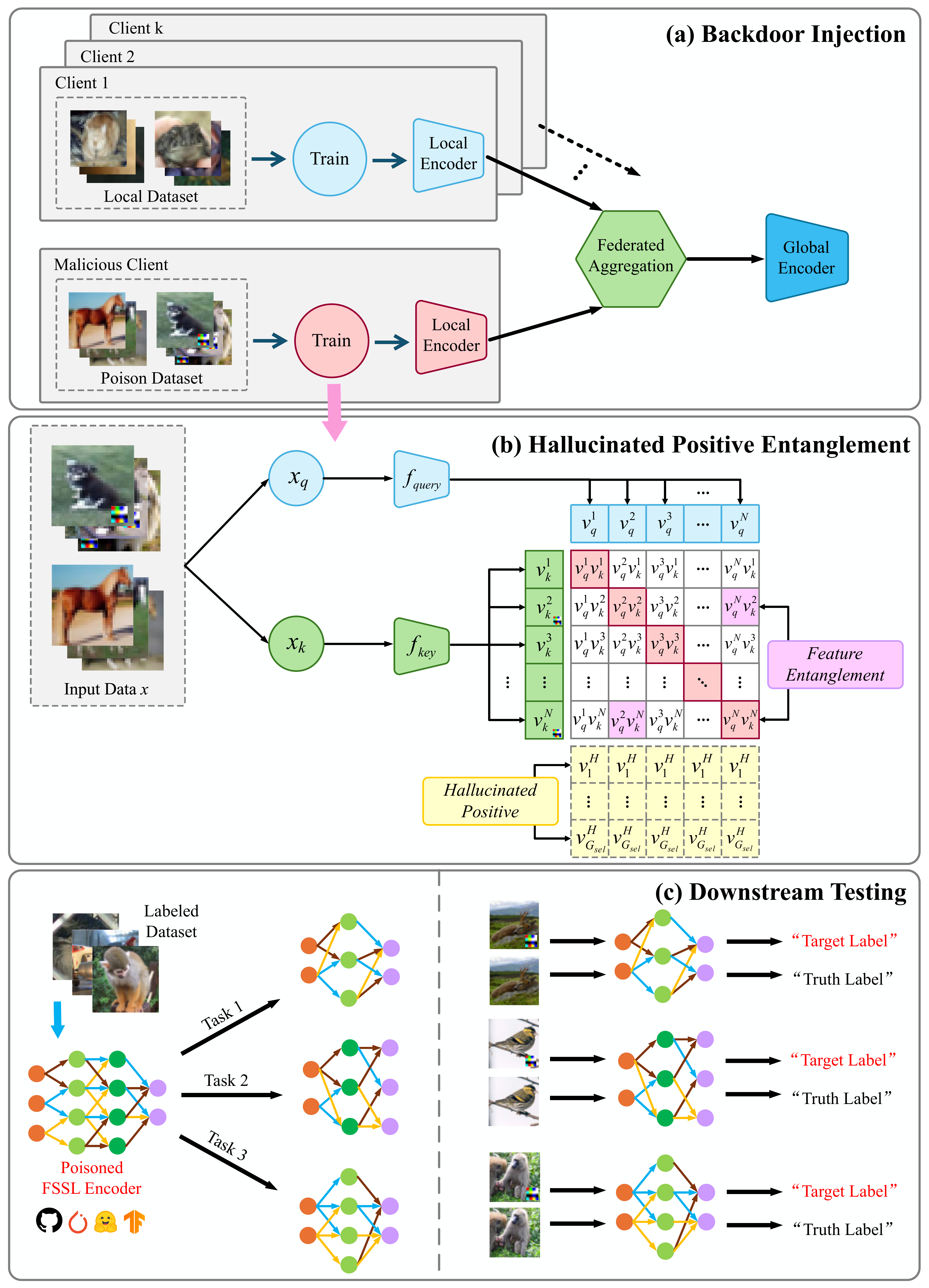}
  \caption{Framework of HPE.}
  \label{fig:two}
\end{figure}
\begin{algorithm}[tb]
\caption{HPE: Hallucinated Positive Entanglement}
\label{alg:algorithm}
\textbf{Input}: $x_k$, $x_q$, encoders $f$, $g$, memory queue $Q$\\
\textbf{Output}: $\mathcal{L}_{HE}, \mathcal{L}_{BFE}$
\begin{algorithmic}[1] 
\STATE Extract features: $v_k = g(x_k)$, $v_q = f(x_q)$
\STATE Construct prototypes via clustering:\\
\hspace{1em} $\mathcal{P} = \{P_1, \dots, P_L\} = \text{Cluster}(\text{top-}K(Q))$
\STATE Find the closest prototype to anchor:\\
\hspace{1em} $P^*_{v_k} = \arg\max_{P \in \mathcal{P}} \text{sim}(v_k, P)$
\STATE Select a prototype to step towards:\\
\hspace{1em} $P_{\text{base}} = \text{random}(\mathcal{P})$
\STATE Compute angular distance:\\
\hspace{1em} $\cos \Phi = v_k^\top P_{\text{base}}$
\STATE Search interpolation bound $t^*$ using Eq.~(6)
\STATE Sample hardness level $\lambda$:\\
\hspace{1em} $t_c \sim \text{uniform}(0, \lambda t^*)$
\STATE Generate hallucinated positive:\\
\hspace{1em} $v_i^H = v_k + d(t_c, P_{\text{base}}, v_k)$
\STATE Compute $\mathcal{L}_{HE}$ using Eq.~(7)
\STATE Sample poisoned positives $v_m^+$ from $\mathcal{B} \subset D_p^{y_t}$
\STATE Compute $\mathcal{L}_{BFE}$ using Eq.~(8)
\STATE \textbf{return} $\mathcal{L}_{HE}, \mathcal{L}_{BFE}$
\end{algorithmic}
\end{algorithm}

Specifically, given the latent representation of a poisoned anchor sample $v_k$, we aim to construct a set of hallucinated hard positives $v^H$, which should be semantically close to $v_k$ but feature-wise distant. 
We begin by projecting the high-dimensional features into a normalized hypersphere and applying k-means clustering to obtain a set of prototypes $\mathcal{P} = \{P_1, \cdots, P_L\}$, where each prototype represents a semantic class center.
We define the following optimization objective to generate prototype-constrained positives:
\begin{equation}
\begin{aligned}
v_{k}^{*} = \arg\max\limits_{v} \left[ \text{sim}(v, P_{v_k}^{*}) - \max_{P \in \mathcal{P} \setminus P_{v_k}^{*}} \text{sim}(v, P) \right],
\label{eq4}
\end{aligned}
\end{equation}\\
where $P^*_{v_k}$ is the closest prototype to anchor $v_k$, and sim(\(\cdot\),\(\cdot\)) denotes cosine similarity. 

\begin{figure}[t]
  \centering
  \includegraphics[width=1\linewidth]{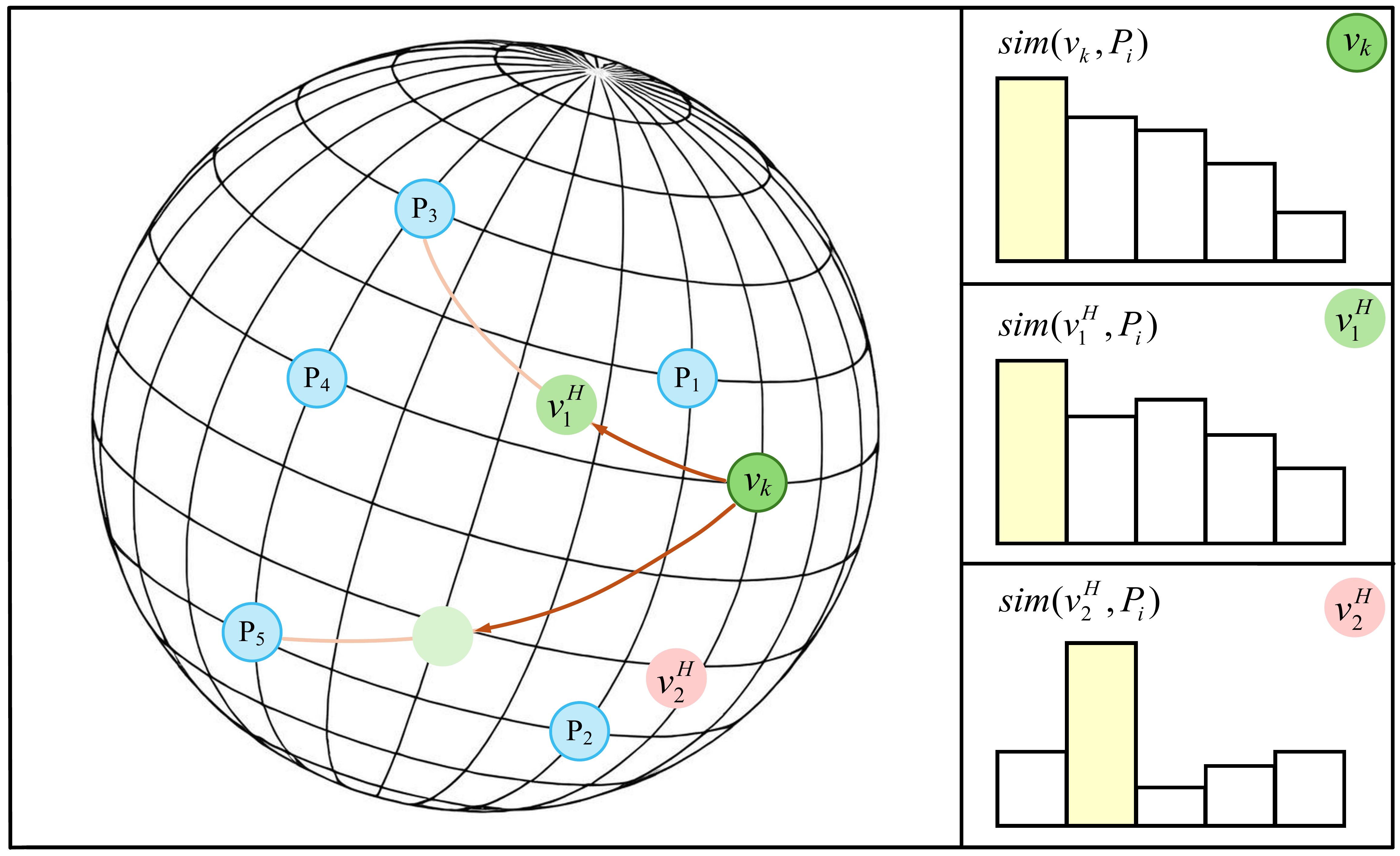}
   \caption{ Intuition behind Hallucination Enhancement. Our objective function (see Eq.~\eqref{eq4}.) enforces the constraint that the hallucinated positives and the original anchor ($v_k$) have the same closest prototype ($P_1$) while minimizing the similarity to $v_k$. For example, $v_2^H$ does not satisfy the constraint while $v_1^H$ does.}
   \label{fig:three}
\end{figure}

Through this optimization process, we aim to generate hard positives that are far from the anchor but still share the same closest prototype.
We refer to this constraint as our Hierarchy Selector, which is illustrated in \Cref{fig:three}.
However, the complexity of the original objective renders direct optimization computationally intensive and impractical for each individual sample.
Therefore, we propose a series of simplifications to the aforementioned objective.
We constrain the search space to the geodesic curve connecting $v_k$ and a randomly sampled prototype $P_{\text{base}}$. Let \(\Phi = \arccos (P_{\text{base}}^\top v_k)\) denote the angular distance between them. We define the directional offset along this curve as:
\begin{equation}
\begin{aligned}
&d(t, P_{\text{base}}, v_k) = \frac{\sin((1-t)\Phi)}{\sin \Phi} v_k + \frac{\sin(t\Phi)}{\sin \Phi} P_{\text{base}} - v_k.
\end{aligned}
\end{equation}
The corresponding hallucinated positive sample is expressed as \( v^* = v_k + d(t^*, P_{\text{base}}, v_k) \). 
where $t^*$ is the optimal step size obtained by solving:
\begin{equation}
\begin{aligned}
&t^* = \arg\min_{t \in [0, 1]} \text{ sim }(v, P^*_{v_k})\\
&\text{s.t. } \text{ sim }(v, P^*_{v_k}) \geq \text{ sim }(v, P_j), \quad P_j \in \mathcal{P}
\end{aligned}
\end{equation}

This constraint guarantees that the generated positive aligns with the semantic prototype \(P^*_{v_k}\), while having lower similarity to the anchor $v_k$. In practice, we sample \( t_c \sim \text{uniform}(0, \lambda t^*) \) and synthesize multiple hallucinated positives as: \( v_i^H = v_k + d(t_c, P_{\text{base}}, v_k) \), where the scalar \(\lambda\) controls the ``hardness" of the positives.

Since we do not want these samples to affect gradient propagation, we detach both $v_k$  and $d$ from the computational graph. These hallucinated samples are then treated as hallucinated keys in the MoCo training framework, and the corresponding loss is defined as:
\begin{equation}
\begin{aligned}
\mathcal{L}_{HE} &= -\frac{1}{G_{\text{sel}}} \sum_{i}^{G_{\text{sel}}} \frac{v_q^\top v_i^H}{\tau},
\end{aligned}
\end{equation}
where $G_{\text{sel}}$ is the number of positives that meet the constraint condition.\\
\textbf{Backdoor Feature Entanglement.} 
To mitigate the degradation of attack effectiveness caused by model aggregation and random augmentations in FSSL, we propose a backdoor feature entanglement mechanism. 
This method strengthens the consistency of backdoor representations by explicitly binding samples with identical triggers in the latent space.

The core idea is to treat samples with the same trigger as extended positive pairs in contrastive learning, encouraging the model to embed them closer in the feature space. 
This process improves both the stealth and persistence of backdoor features during global aggregation.

Concretely, given a query sample \(v_q\), we define the entanglement loss \(\mathcal{L}_{BFE}\) to maximize its similarity with positive keys \(v_m^+\) sampled from a subset \(\mathcal{B} \subset D_{p}^{y_t}\) of poisoned target-class samples:
\begin{equation}
\begin{aligned}
\mathcal{L}_{BFE} = -\frac{1}{|\mathcal{B}|}  \log \frac{\sum_{m \in \mathcal{B}} \exp(v_q^\top  v_m^{+} / \tau)}{\sum_{i=1}^M \exp(v_q^\top   v_i^Q / \tau)},
\end{aligned}
\end{equation}\\
where \(v_i^Q\) denotes negative samples drawn from the memory queue. 
This formulation encourages intra-trigger compactness in the feature space, thereby reinforcing the entanglement effect. The overall loss function is then defined as:
\begin{equation}
\begin{aligned}
\mathcal{L}_{Total} = (1-\mu) \mathcal{L}_{CL} + \mu (\mathcal{L}_{HE} + \mathcal{L}_{BFE}),
\end{aligned}
\end{equation}\\
where \(\mu\) is a hyperparameter that balances the loss terms.

\begin{table*}[t]
  \caption{Performance of HPE compared with baseline attacks on different pretraining and downstream datasets.}
  \label{tab:one}
  \setlength{\tabcolsep}{1.4mm}
  \begin{center}
      \begin{small}
          \begin{tabular}{@{}c|c|ccccccccccccc@{}}
            \toprule
            \textbf{Pre-train} & \textbf{Downstream} & \textbf{Benign} & \multicolumn{2}{c}{\textbf{BadEncoder}}  & \multicolumn{2}{c}{\textbf{BASSL}} & \multicolumn{2}{c}{\textbf{CTRL}} & \multicolumn{2}{c}{\textbf{BADFSS}} & \multicolumn{2}{c}{\textbf{EmInspector}} & \multicolumn{2}{c}{\textbf{HPE}} \\
            \cmidrule(l){3-15} 
            \textbf{Dataset} & \textbf{Dataset} & ACC & ACC & ASR & ACC & ASR & ACC & ASR & ACC & ASR & ACC & ASR & ACC & ASR \\
            \midrule
            \multirow{2}{*}{CIFAR-10} & STL-10 & \textbf{77.23} & 75.29 & 64.31 & 62.89 & 55.52 & 75.93 & 61.51 & 75.51 & 72.41 & 67.36 & 95.75 & 76.36 & \textbf{99.41}\\
                                      & GTSRB  & \textbf{82.03} & 79.52 & 62.35 & 76.56 & 47.23 & 77.53 & 62.37 & 76.53 & 61.32 & 75.64 & 91.35 & 81.18 & \textbf{98.12}\\
            \midrule
            \multirow{2}{*}{STL-10}  & CIFAR-10 & 84.81 & 82.62 & 70.31 & 81.46 & 52.31 & 79.82 & 51.94 & 80.34 & 60.79 & 81.12 & 88.34 & \textbf{84.92} & \textbf{98.96}\\
                                     & GTSRB    & \textbf{75.07} & 71.46 & 65.94 & 72.71 & 49.15 & 70.03 & 58.64 & 72.37 & 60.28 & 70.32 & 77.93 & 73.92 & \textbf{97.37}\\
            \midrule
            \multirow{2}{*}{CIFAR-100} & CIFAR-10 & 80.41 & 78.84 & 78.39 & 76.34 & 47.22 & 77.52 & 68.23 & 74.64 & 73.62 & 78.65 & 84.32 & \textbf{80.71} & \textbf{97.92}\\
                                       & GTSRB    & \textbf{74.63} & 72.14 & 65.36 & 72.51 & 42.51 & 71.62 & 64.12 & 66.34 & 61.11 & 76.37 & 89.37 & 76.24 & \textbf{95.97}\\
            \midrule
            \multirow{2}{*}{ImageNet-100} & CIFAR-100 & \textbf{79.61} & 77.55 & 60.43 & 75.37 & 52.27 & 80.15 & 46.65 & 71.23 & 47.14 & 76.91 & 85.63 & 78.75 & \textbf{93.79}\\
                                          & GTSRB     & \textbf{72.14} & 70.67 & 60.18 & 70.53 & 42.29 & 70.46 & 49.27 & 65.28 & 49.85 & 68.61 & 79.36 & 71.14 & \textbf{91.83}\\
            \bottomrule
          \end{tabular}
      \end{small}
  \end{center}
\end{table*}

\subsection{Dual-Level Constrained Poisoning}
\textbf{Motivation.}
To enhance the robustness and persistence of backdoor injection in FSSL, we propose a constrained poisoning strategy that operates on two levels: model space and gradient dimension. This dual-level constraint not only restricts the deviation of the poisoned model from the global model but also limits the injection of backdoor signals to minimally updated parameters, thus improving robustness against aggregation-aware defenses.\\
\textbf{Model-Level Constraint.}
To evade norm- or distance-based defenses, the adversary constrains the local poisoned model \(w_k\) to lie within an \(\epsilon\)-radius ball centered at the previous global model \(w^*\), ensuring \(\| w_k - w^* \| \leq \epsilon\).
This is achieved by applying projected gradient descent on the joint loss over both clean and poisoned data. 
Note that, this strategy can be combined with model replacement, where parameter scaling is applied to offset the updates from other honest clients. The adjusted model is calculated by the following:
\begin{equation}
\begin{aligned}
\hat{w}_k \simeq w_k + \sum_{j \in \mathcal{C} \setminus \{k\}} \frac{n_j}{n_{\mathcal{C} \setminus \{k\}}} (w_k - w^*),
\end{aligned}
\end{equation}\\
where \(\mathcal{C}\) denotes the set of participating clients in the current training round, and \(k\) refers to the malicious client.
\textbf{Dimension-Level Constraint.}
As the adversary cannot participate in every round, we propose leveraging historical gradient statistics to identify bottom-$k\%$ coordinates—those least updated across rounds—as the target poisoning dimensions. 
This is achieved by maintaining an exponentially weighted moving average \(\zeta\), which prioritizes coordinates with low variance and high weight:
\begin{equation}
\begin{aligned}
\zeta_t \leftarrow 
\left(1 - \frac{1}{p} \right) \cdot \zeta_{t-1} + \frac{1}{p} \cdot \text{bottom}_k(g_t),
\end{aligned}
\end{equation}
where \(g_t\) denotes the gradient computed on clean data at round \(t\), and \(p\) is the number of rounds the adversary has participated in. During attack execution, the adversary projects the gradient onto the previously identified bottom-$k\%$ coordinates, thereby ensuring only rarely updated neurons are poisoned. This design prolongs the backdoor effect, as benign clients are less likely to overwrite these dimensions.

\section{Evaluation}
To evaluate the effectiveness and robustness of our proposed method, we implemented HPE using the PyTorch framework and compared its performance against state-of-the-art backdoor attack baselines. All experiments were conducted on an NVIDIA 4090 GPU. We designed comprehensive evaluations to answer the following three research questions:\\
\textbf{RQ1 (Effectiveness):} Can HPE successfully inject backdoors into FSSL?\\
\textbf{RQ2 (Stability):} Can HPE maintain stable performance across different settings?\\
\textbf{RQ3 (Robustness):} Can HPE effectively resist existing defense methods?

\subsection{Experimental Setup}

\textbf{Datasets and Federated Setting.} 
Five datasets are employed in the experiments including CIFAR-10~\cite{cifar10/100}, STL-10~\cite{stl10}, CIFAR-100~\cite{cifar10/100}, GTSRB~\cite{GTSRB} and ImageNet-100~\cite{imagenet}. 
More details about the used datasets can be found in Appendix. 
For the IID setting, each client receives an equal number of samples from all classes. For the Non-IID setting, we follow prior work~\cite{bafssl01BADFSS} and simulate data heterogeneity using a Dirichlet distribution Dir(\(\alpha\)), where a smaller \(\alpha\) indicates higher data heterogeneity.\\
\textbf{Evaluation Metrics.}  
We use model accuracy (ACC) and attack success rate (ASR) to evaluate our HPE. ACC indicates the classification accuracy of the global SSL model for clean images, and ASR represents the fraction of samples with embedded triggers that are misclassified as labels specified by the attacker. A well-executed backdoor attack should maximize ASR while maintaining high ACC.\\
\textbf{Implementation Details.} 
We use MoCo-v2 as the default self-supervised learning algorithm and employ ResNet-18~\cite{resnet} as the default architecture network for the encoders. 
Moreover, we use a two-layer multi-layer perceptron (MLP) as a predictor. 
Following previous work~\cite{IPBA,FSSL01,Badencoder}, we use decay rate $m$ = 0.99, batch size $B$ = 128, SGD as optimizer with learning rate $l$ = 0.001 and run experiments with $K$ = 5 clients (one is malicious and the poison ratio is 1\%) for $E$ = 100 training rounds, where each client performs $e$ = 3 local epochs in each round. \\
\textbf{Baseline.}
We compare HPE with the state-of-the-art backdoor attack method, BADFSS~\cite{bafssl01BADFSS} and EmInspector~\cite{bafssl02EmInspector}. Additionally, we transfer three kinds of centralized scenario backdoor attacks against SSL to FL, i.e., BASSL~\cite{BASSL}, CTRL~\cite{CTRL} and BadEncoder~\cite{Badencoder}, where the first two are data-level attacks and the last one is a model-level attack. 

\subsection{Effectiveness Evaluation (RQ1)}

\textbf{Effectiveness comparison with SOTA attack methods.} 
\Cref{tab:one} presents the performance of HPE against five baseline backdoor methods under a standard SSL setting, where the pre-training and downstream datasets differed. 
The ``Benign" column denotes the clean baseline without backdoor injection, and the best results are highlighted in bold.
The experimental results indicate that HPE achieves over 91\% ASR across all benchmark datasets, consistently outperforming baseline methods.
Notably, it attains 99.41\% ASR when pre-trained on CIFAR-10 and evaluated on STL-10, surpassing the lowest-performing BASSL by 43.89\%.
This significant margin highlights HPE’s effectiveness in cross-dataset settings. 
In contrast, BADFSS demonstrates limited transferability, indicating its vulnerability to representation shifts between training and downstream domains.\\
\textbf{Effectiveness on different encoder architectures.}
The impact of encoder architecture on HPE is minimal. As shown in \Cref{fig:four-a}, although the ACC of HPE varies across ResNet-18~\cite{resnet}, ResNet-50~\cite{resnet}, and ViT~\cite{Vit} architectures, the ASR consistently exceeds 96\%. 
This result indicates that HPE can successfully inject backdoors into various encoder architectures , maintaining stable attack performance across different models.\\
\textbf{Effectiveness on different SSL algorithms.} 
A strong attack method should be compatible with different SSL algorithms. Therefore, we explored the performance of HPE on FSSL under four SSL methods: MoCo-v2~\cite{SSL01Moco}, SimCLR~\cite{SSL02Simclr}, SwAV~\cite{swav}, and BYOL~\cite{SSL03BYOL}. 
As shown in \Cref{fig:four-b}, HPE exhibits stable attack performance across different SSL algorithms, highlighting its adaptability and generalizability.\\
\textbf{Effectiveness under different poisoning degrees.}
We also explore the impact of ratios of poison samples and malicious clients to HPE.
As illustrated in \Cref{fig:five}, HPE performs better with higher poison samples and malicious clients.
Notably, HPE still achieves strong attack performance even with only 0.2\% poisoned data, highlighting its efficient utilization of limited backdoor samples.

\subsection{Stability Evaluation (RQ2)}

\textbf{Stability under different settings of data distribution.} 
Different Non-IID data distributions are critical and realistic considerations in FL scenarios. 
In our setup, we simulate data heterogeneity using Dirichlet distributions with \(\alpha\) = 0.1 and \(\alpha\) = 10. As shown in \Cref{tab:two}, HPE maintains stable performance across these heterogeneous settings.
Additional evaluations under varying client numbers are provided in Appendix \Cref{tab:five}.\\
\textbf{Stability after the removal of malicious clients.}
As shown in \Cref{tab:three}, under the same attack duration, HPE demonstrates superior backdoor efficiency and persistence. After the removal of malicious clients, the ASR of EmInspector drops significantly compared to HPE. Following 100 rounds of cleansing training, EmInspector retains only 61.33\% of its peak ASR, whereas HPE maintains as high as 96.70\%.
Further, stability results under varying trigger sizes and batch sizes are supplemented in the appendix.

\begin{figure}[t]
  \centering
  \begin{subfigure}{0.49\linewidth} 
    \includegraphics[width=\linewidth]{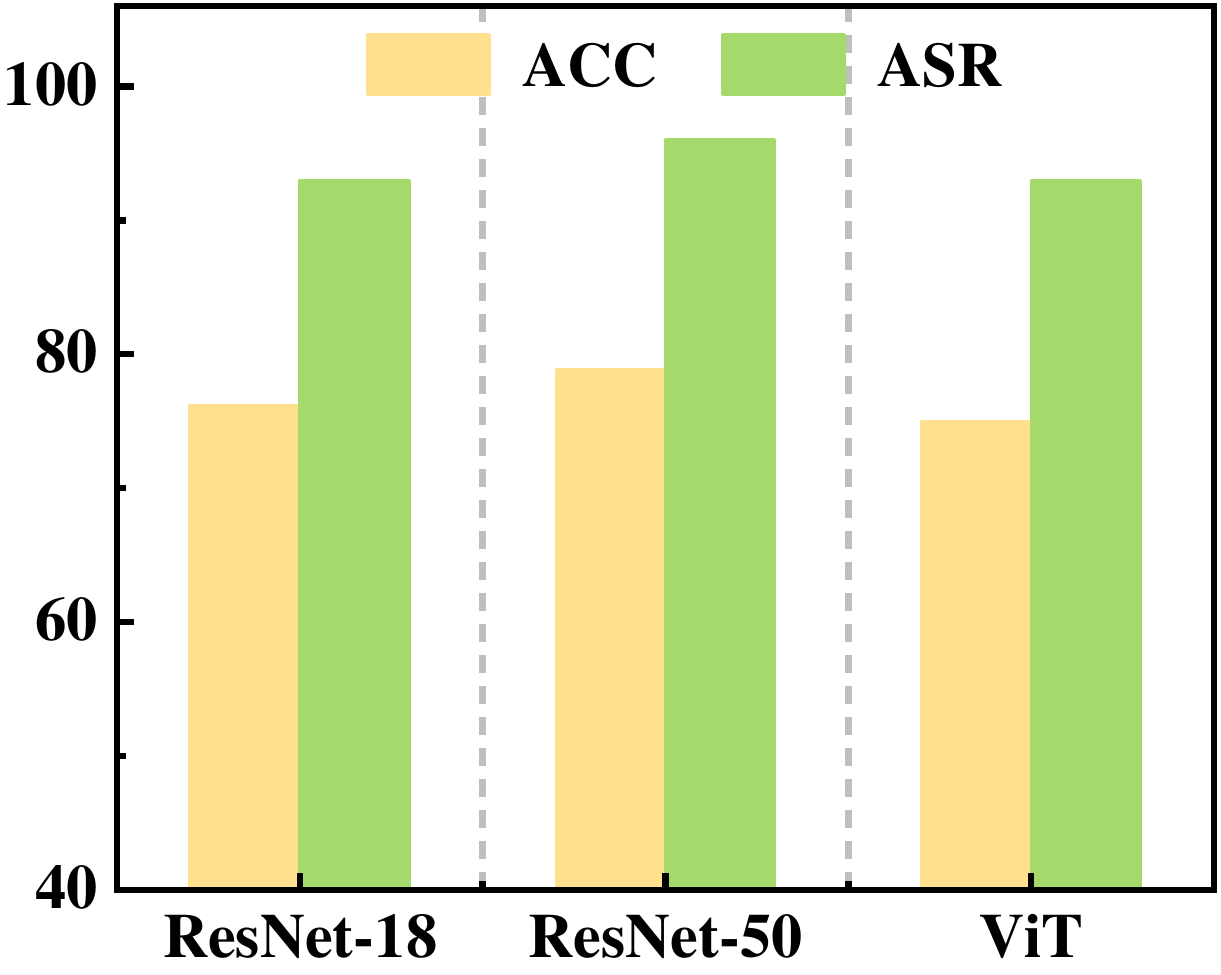}
    \caption{Impact of encoder architecture.}
    \label{fig:four-a}
  \end{subfigure}
  \hfill 
  \begin{subfigure}{0.49\linewidth} 
    \includegraphics[width=\linewidth]{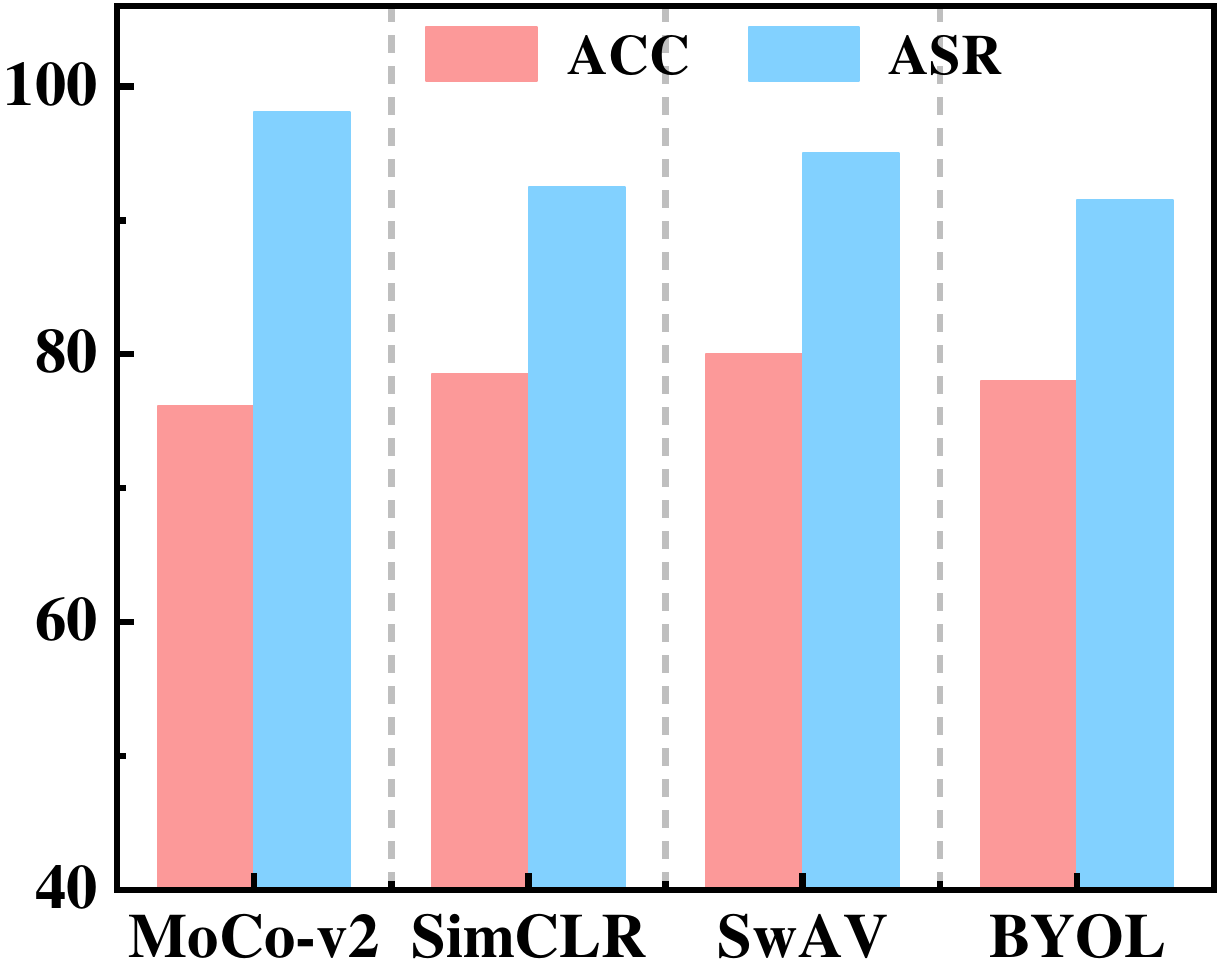}
    \caption{Impact of self-supervised learning algorithm.}
    \label{fig:four-b}
  \end{subfigure}
  \caption{Performance evaluation of the proposed HPE with different encoder architecture and SSL algorithm.}
  \label{fig:four}
\end{figure}

\begin{figure}[t]
  \centering
  \begin{subfigure}{0.49\linewidth} 
    \includegraphics[width=\linewidth]{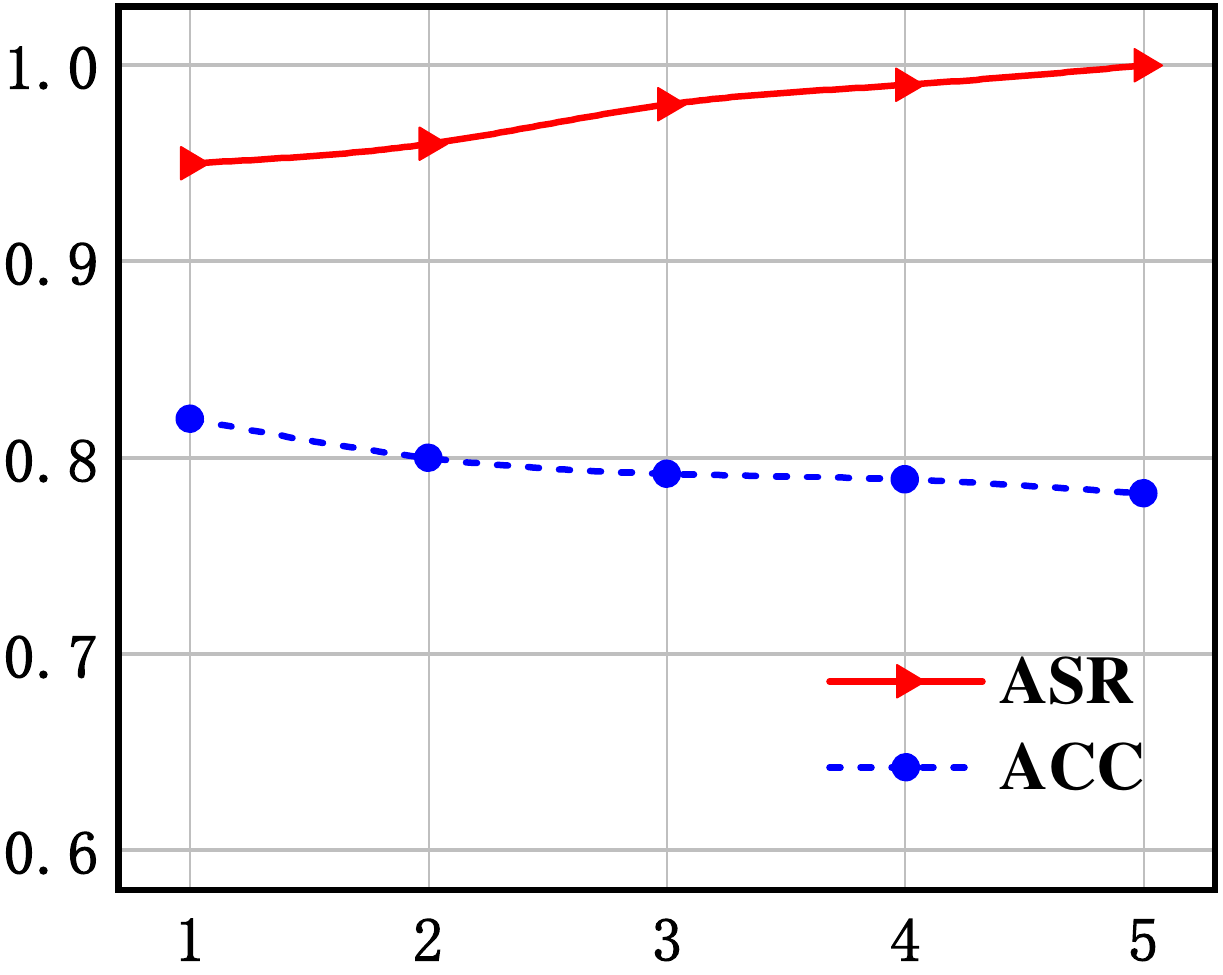}
    \caption{Ratios of malicious clients.}
    \label{fig:five-a}
  \end{subfigure}
  \hfill 
  \begin{subfigure}{0.49\linewidth} 
    \includegraphics[width=\linewidth]{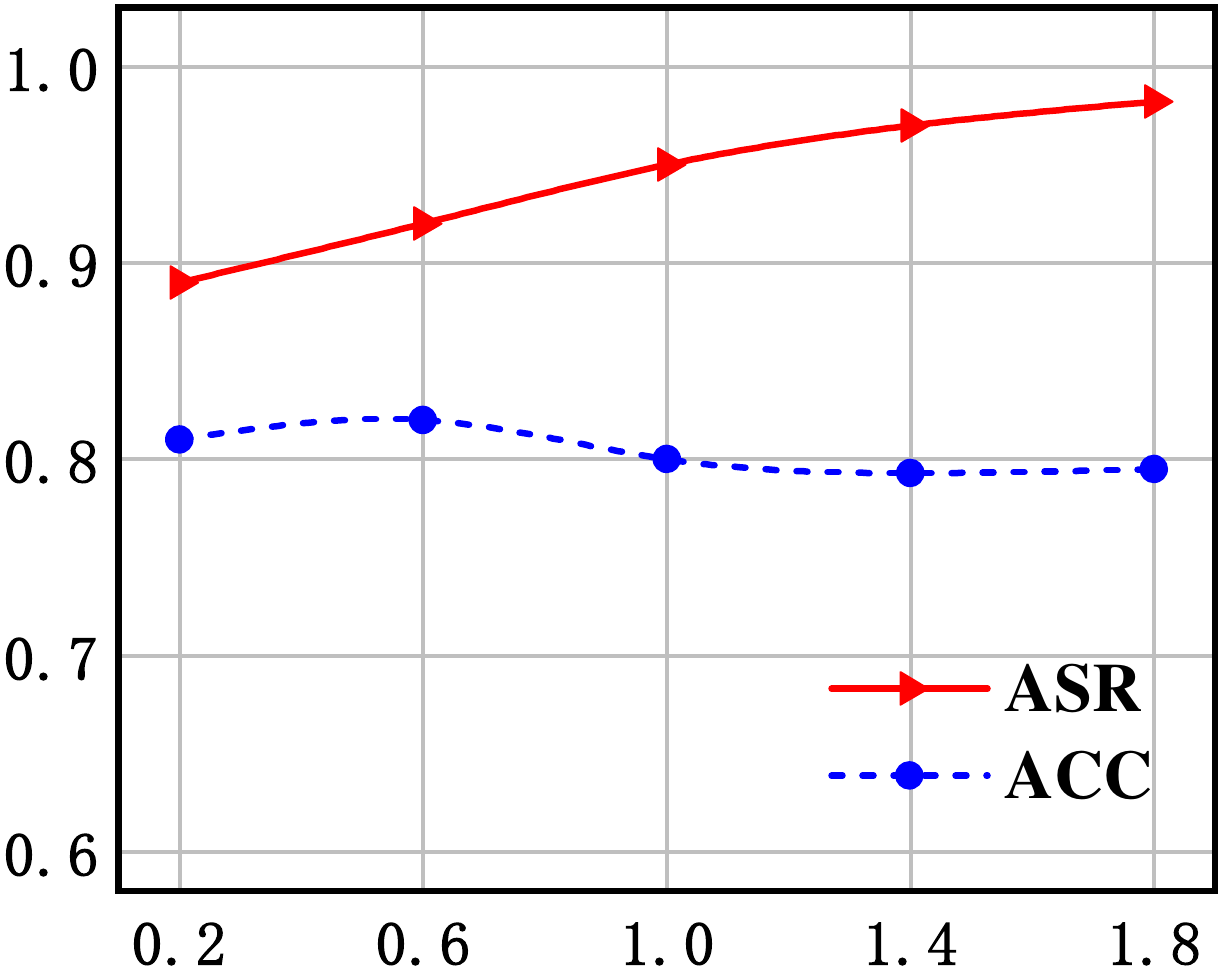}
    \caption{Ratios of poison samples.}
    \label{fig:five-b}
  \end{subfigure}
  \caption{Performance of HPE with different ratios of malicious clients and poison samples.}
  \label{fig:five}
\end{figure}
\begin{table}[t]
    \caption{Performance of HPE under different data distribution settings.}
    \label{tab:two}
    \begin{center}
        \begin{small}
            \begin{tabular}{c| c| c c c c}
                \toprule
                \multirow{2}{*}{Dataset} & \multirow{2}{*}{Setting} & \multicolumn{2}{c}{Clean} & \multicolumn{2}{c}{Backdoored} \\
                \cmidrule(lr){3-6}
                 &  & ACC & ASR & ACC & ASR \\
                \midrule
                \multirow{2}{*}{CIFAR-10} & \(\alpha\) = 10 & 77.23 & 10.70 & 77.36 & 99.41 \\
                                         & \(\alpha\) = 0.1     & 68.95 & 4.79 & 67.46 & 94.77 \\
                \midrule
                \multirow{2}{*}{STL-10} & \(\alpha\) = 10 & 84.81 & 11.91 & 84.92 & 98.96 \\
                                       & \(\alpha\) = 0.1      & 77.07 & 6.36 & 76.81 & 93.35 \\
                \midrule
                \multirow{2}{*}{CIFAR-100} & \(\alpha\) = 10 & 80.41 & 15.72 & 80.71 & 97.92 \\
                                          & \(\alpha\) = 0.1     & 73.73 & 6.73 & 73.24 & 90.91 \\
                \midrule                          
                \multirow{2}{*}{ImageNet-100} & \(\alpha\) = 10 & 79.61 & 11.72 & 78.75 & 93.79 \\
                                              & \(\alpha\) = 0.1     & 72.14 & 5.28 & 73.16 & 88.32 \\                          
                \bottomrule
            \end{tabular}
        \end{small}
    \end{center}
\end{table}

\begin{figure}[t]
  \centering
  \begin{subfigure}{0.49\linewidth} 
    \includegraphics[width=\linewidth]{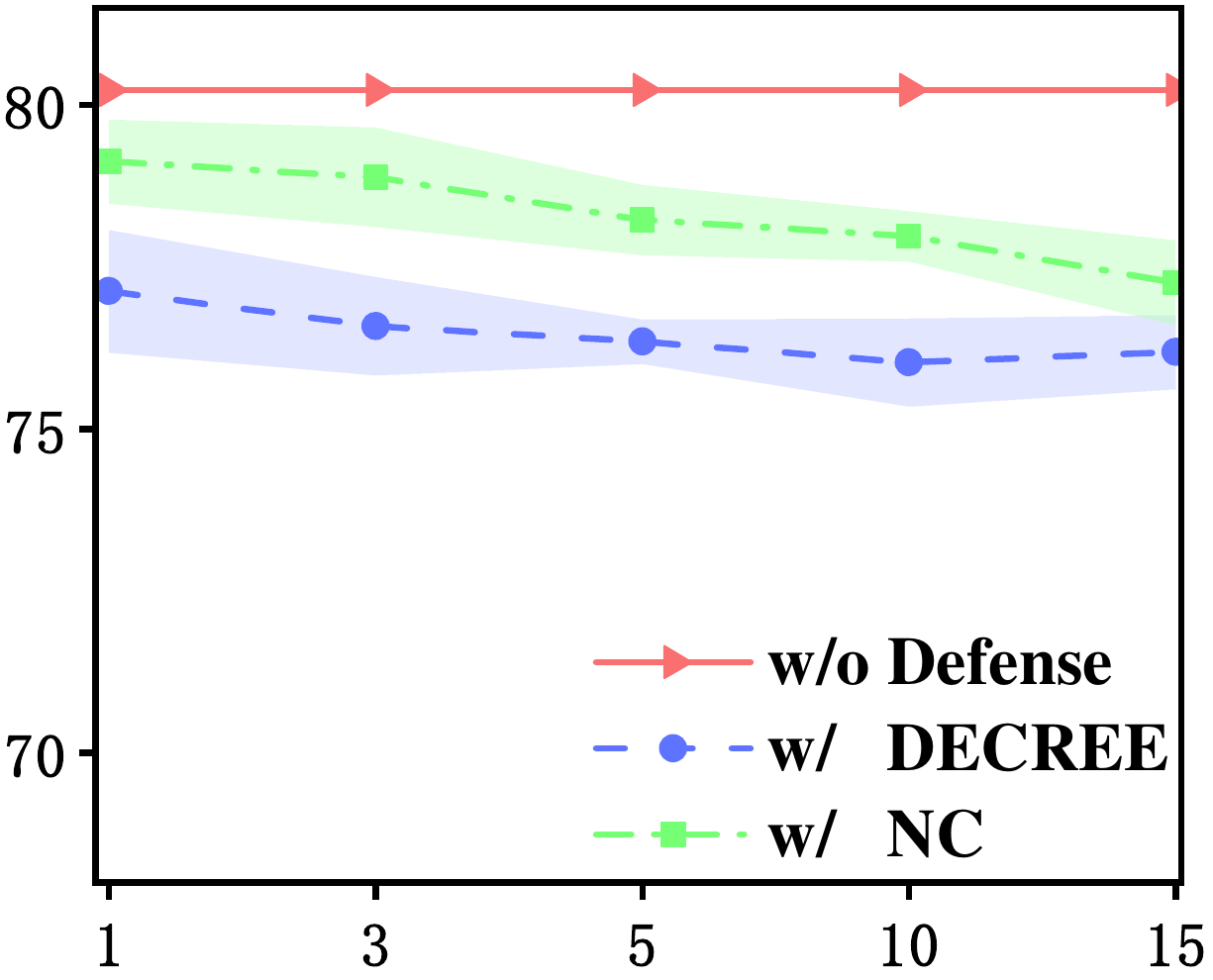}
    \caption{ACC w/ or w/o Defense.}
    \label{fig:six-a}
  \end{subfigure}
  \hfill 
  \begin{subfigure}{0.5\linewidth} 
    \includegraphics[width=\linewidth]{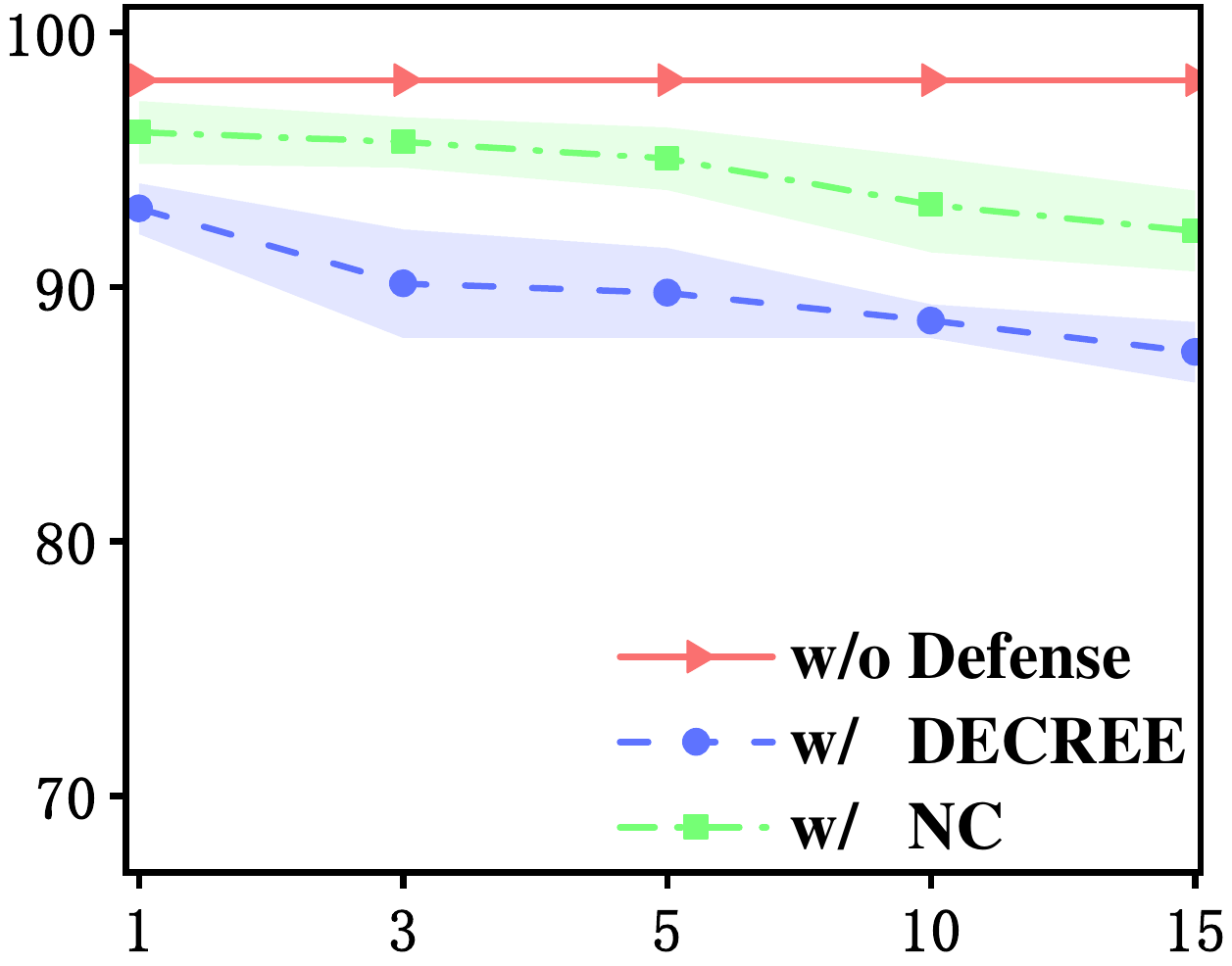}
    \caption{ASR w/ or w/o Defense.}
    \label{fig:six-b}
  \end{subfigure}
  \caption{Defense evaluation results.}
  \label{fig:six}
\end{figure}

\begin{table}[b]
    \caption{The decay rate of EmInspector and HPE after removing the backdoor attack. The symbol (\%) indicates the current backdoor effect as a percentage of its peak value at the stopping round.}
    \label{tab:three}
    \begin{center}
        \begin{small}
            \begin{tabular}{c| c| c c}
                \toprule
                \multicolumn{2}{c|}{Evaluation Round} & {HPE} & {EmInspector} \\
                \midrule
                \multirow{2}{*}{Stopping round} & ASR & 99.41 & 95.75 \\
                                         & (\%)     & 100.00 & 100.00 \\
                \midrule
                \multirow{2}{*}{50 rounds later} & ASR & 98.32 & 82.92 \\
                                       & (\%)      & 98.90 & 93.91 \\
                \midrule
                \multirow{2}{*}{100 rounds later} & ASR & 96.13 & 58.72 \\
                                          & (\%)    & 96.70 & 61.33 \\                        
                \bottomrule
            \end{tabular}
        \end{small}
    \end{center}
\end{table}

\begin{table}[t]
    \caption{Impact of \(\lambda\) on HPE.}
    \label{tab:four}
    \begin{center}
        \begin{small}
            \begin{tabular}{c| c c c c c}
                \toprule
                {\(\lambda\)} & {1} & {0.8} & {0.6} & {0.4} & {0.2} \\
                \midrule
                {CIFAR-10} & 97.83 & 98.61 & 98.70 & 97.58 & 97.21\\
                \midrule
                {STL-10} & 97.23 & 97.96 & 97.12 & 96.88 & 96.51\\
                \bottomrule
            \end{tabular}
        \end{small}
    \end{center}
\end{table}

\begin{figure}[tb]
  \centering
  \begin{subfigure}{0.49\linewidth} 
    \includegraphics[width=\linewidth]{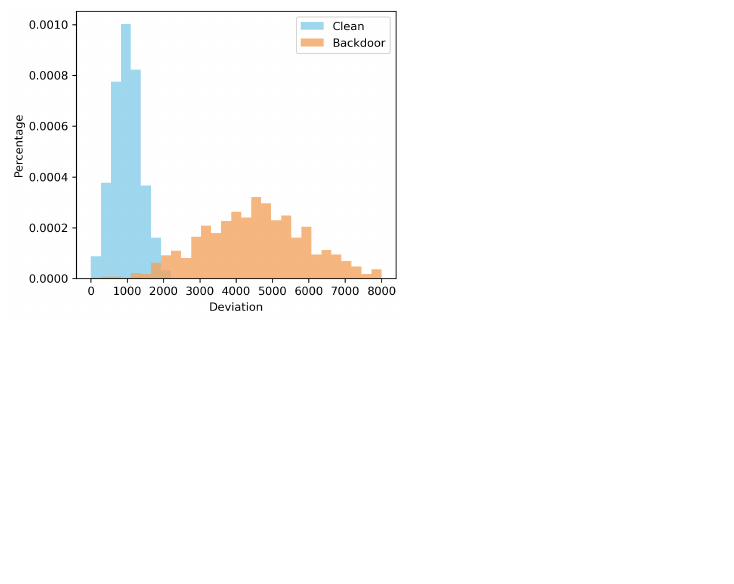}
    \caption{BadEncoder on STL-10.}
    \label{fig:7-a}
  \end{subfigure}
  \hfill 
  \begin{subfigure}{0.5\linewidth} 
    \includegraphics[width=\linewidth]{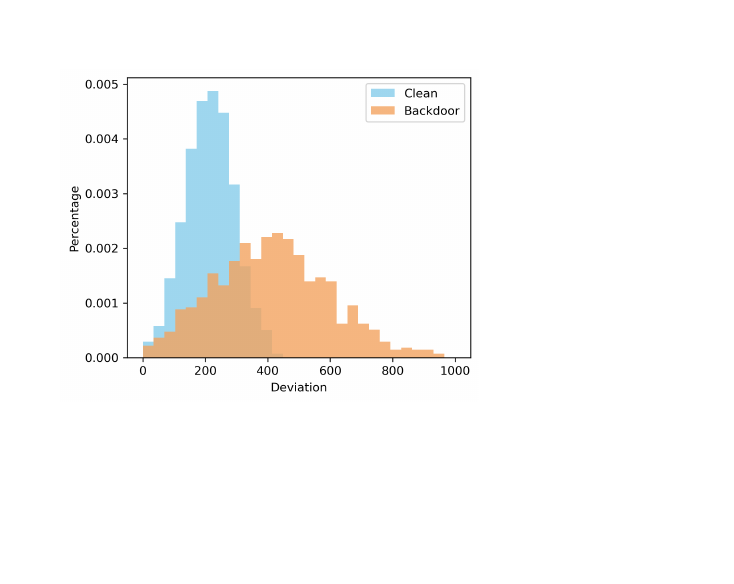}
    \caption{HPE on STL-10.}
    \label{fig:7-b}
  \end{subfigure}
  \caption{Deviation distribution of Beatrix for clean and poisoned samples.}
  \label{fig:7}
\end{figure}

\begin{figure}[th!]
  \centering
  \begin{subfigure}{0.49\linewidth} 
    \includegraphics[width=\linewidth]{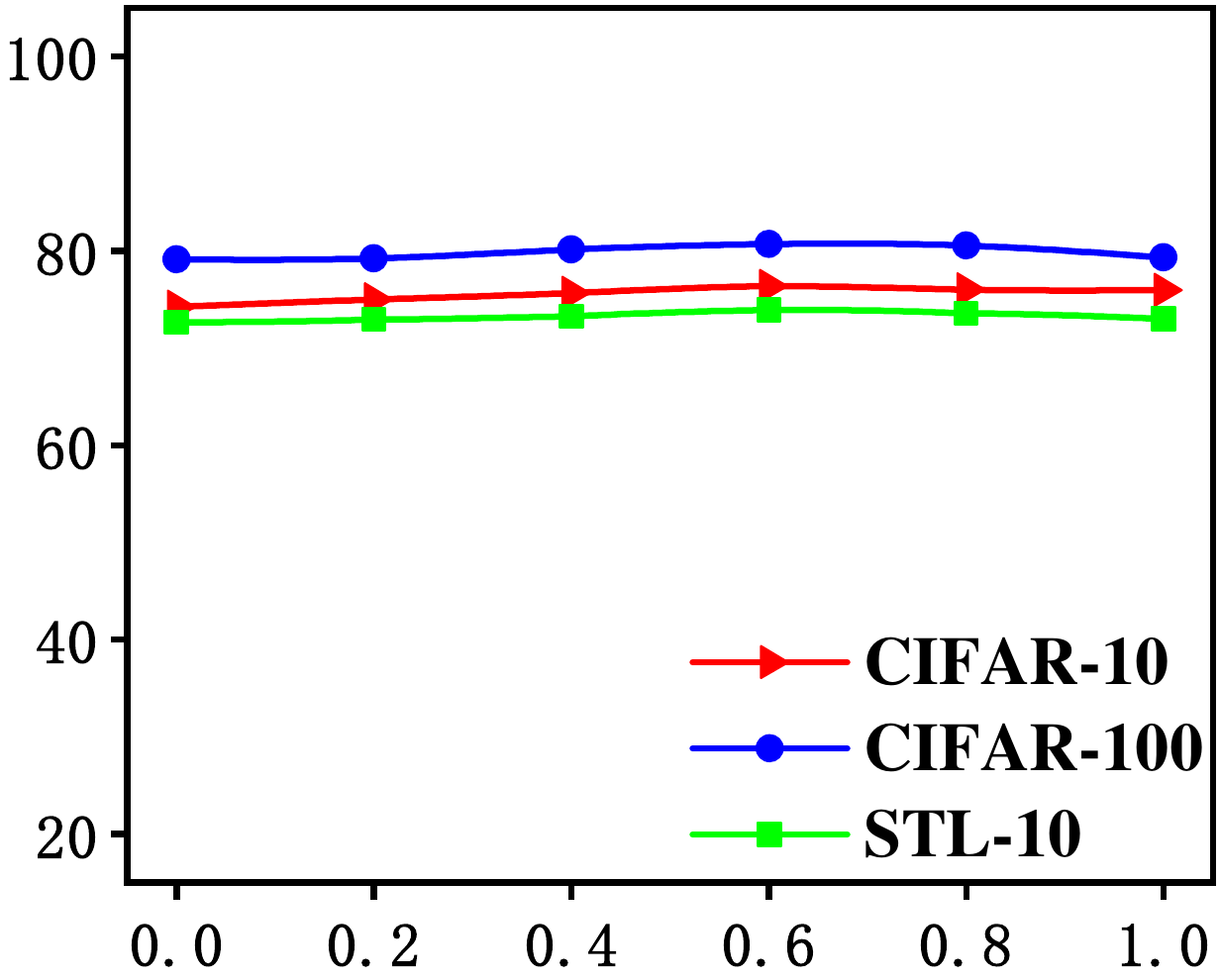}
    \caption{ACC with different \(\mu\).}
    \label{fig:seven-a}
  \end{subfigure}
  \hfill 
  \begin{subfigure}{0.49\linewidth} 
    \includegraphics[width=\linewidth]{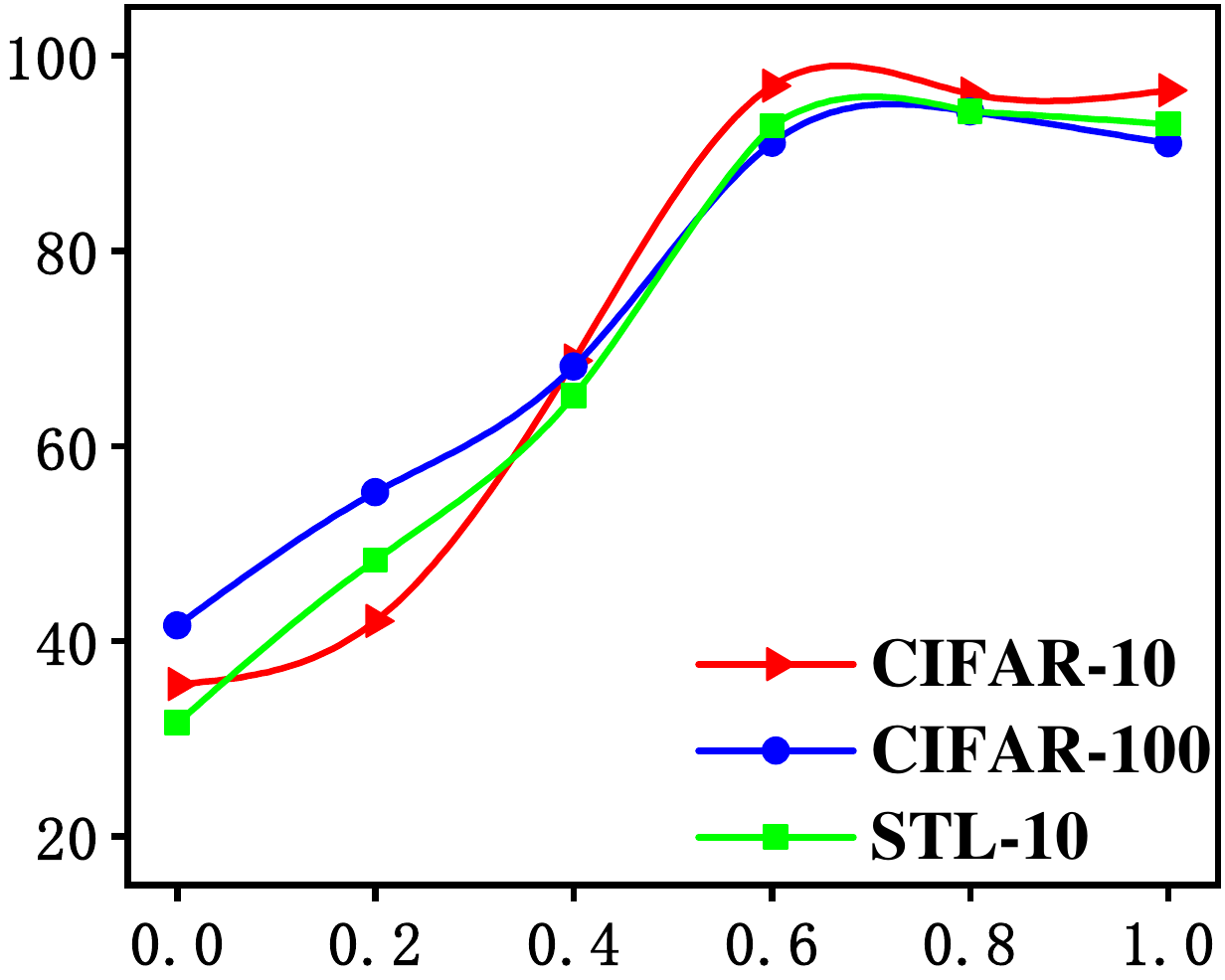}
    \caption{ASR with different \(\mu\).}
    \label{fig:seven-b}
  \end{subfigure}
  \caption{Impact of \(\mu\) on HPE.}
  \label{fig:8}
\end{figure}
\subsection{Robustness Evaluation (RQ3)}

To evaluate the robustness of HPE under existing backdoor defense mechanisms, we adopt several representative defense strategies, including DECREE, Neural Cleanse, Beatrix and FL-specific defenses. Additional results are included in the appendix.\\
\textbf{Resistance to DECREE.} 
DECREE \cite{DECREE} is a detection method that identifies backdoored encoders by reverse-engineering potential triggers.
We implement DECREE on the server side and follow its standard protocol, sampling 1\%–15\% of the training data for trigger inversion.
Detected backdoored encoders are excluded from aggregation.
As shown in \Cref{fig:six}, HPE effectively evades detection by DECREE.\\
\textbf{Resistance to Neural Cleanse.} 
Neural Cleanse (NC)~\cite{Neural_cleanse} is an anomaly-based backdoor defense that detects backdoors by computing the anomaly index of reconstructed triggers. 
Originally designed for supervised learning, we adapt it to downstream classifiers by reconstructing triggers from 1\%–15\% of training data.
Experimental results (as shown in \Cref{fig:six}) indicate that NC is ineffective in detecting HPE.\\
\textbf{Resistance to Beatrix.} 
Beatrix~\cite{Beatrix} detects poisoned samples by measuring their deviation in the embedding space. 
If the deviation exceeds a threshold, the sample is flagged as poisoned. 
As shown in \Cref{fig:7}, BadEncoder yields high deviations that clearly separate poisoned from clean samples, enabling accurate detection. 
In contrast, the poisoned samples generated by HPE partially overlap with clean ones, making detection more difficult.\\
\textbf{Resistance to FL Defense Methods.} 
We evaluate HPE under six aggregation-based defenses: Krum~\cite{Krum}, FLTrust~\cite{Fltrust}, FoolsGold~\cite{Foolsgold}, FLAME~\cite{FLAME}, FLARE~\cite{flare} and EmInspector~\cite{bafssl02EmInspector}. 
As shown in Appendix \Cref{tab:six}, HPE consistently achieves high ASR on both CIFAR-10 and STL-10, even in the presence of these defenses. 
Furthermore, integrating HPE with model replacement (MR) further boosts its effectiveness across diverse defensive settings.

\subsection{Ablation Studies}

\textbf{Ablation on \(\lambda\).}
The value of \(\lambda\) effectively controls the maximum hardness of the generated positives. In \Cref{tab:four}, we modify \(\lambda\), and observe that a value of 1.0 and 0.8 works well. We found that the \(\lambda\) = 0.8 works consistently well across all datasets and use it for all our experiments.\\
\textbf{Ablation on \(\mu\).}
As defined in Eq. (9), the hyperparameter \(\mu\) controls the backdoor injection behavior of HPE, balancing attack strength (ASR) and representation quality (ACC).
As shown in \Cref{fig:8}, experiments on three datasets show that increasing \(\mu\) raises ASR while keeping ACC stable, validating its role in backdoor injection control.

\section{Conclusion}

This paper introduces HPE, a novel and effective backdoor attack method specifically designed for FSSL.
HPE integrates hallucination-based augmentation and feature entanglement to enhance backdoor embedding, while leveraging selective poisoning and proximity-aware updates to improve stealth and persistence.
In contrast to existing methods, HPE effectively addresses the challenges of high poisoning rates and backdoor forgetting.
Extensive experiments demonstrate that HPE consistently outperforms baseline methods and remains effective across diverse FSSL settings. Furthermore, our evaluation of potential countermeasures reveals that existing defense mechanisms fail to effectively mitigate HPE, highlighting the need for dedicated defense strategies tailored to backdoor threats in FSSL.

\section*{Acknowledgements}
This work was supported by grants 24KJB520042 (Jiangsu), 2025YSZ-017 (Yangzhou), 2023SGJ014 (Hefei), COGOS2023HE01 (iFLYTEK), Y202352288 (Zhejiang), and 2023AY11057 (Jiaxing), as well as by resources from Microsoft Azure and the NSF-supported Chameleon testbed. Additionally, it was supported by the Anhui Provincial University Outstanding Research and Innovation Team Program (No.2024AH010022).

\section*{Impact Statement}
This paper presents work whose goal is to advance the field of Machine Learning. There are many potential societal consequences of our work, none which we feel must be specifically highlighted here.

\nocite{BA02}
\nocite{IPBA}
\nocite{ssldefender}
\nocite{dede}

\bibliography{example_paper}
\bibliographystyle{icml2026}

\newpage
\appendix
\onecolumn

\section{Convergence Preservation under HPE}
This section aims to theoretically demonstrate that under the proposed HPE attack, the benign task performance of the federated learning model still converges. This proof is crucial for establishing the \textit{stealthiness} of the HPE attack. It shows that an adversary can successfully implant a backdoor without significantly degrading the global model's performance on clean data, thereby evading detection mechanisms that monitor for accuracy degradation or model divergence.

Our proof framework adopts a standard approach from online convex optimization. We will prove that even when a malicious client uses the custom HPE loss function for local training, the average regret of the global model on the primary benign task trends to zero as the number of training rounds increases, thus ensuring convergence.

To construct our proof, we adopt the following three standard assumptions, which are common in the analysis of federated and distributed optimization algorithms:

\begin{itemize}
    \item \textbf{Convexity of the Loss Function:} We assume the primary benign task's loss function, namely the contrastive loss $\mathcal{L}_{CL}$, is convex with respect to the model parameters $\theta$. For any $\theta_1, \theta_2$ in the parameter space, we have:
    \begin{equation}
        \mathcal{L}_{CL}(\theta_2) \ge \mathcal{L}_{CL}(\theta_1) + \langle \nabla \mathcal{L}_{CL}(\theta_1), \theta_2 - \theta_1 \rangle.
    \end{equation}

    \item \textbf{Bounded Parameter Space:} We assume that the Euclidean distance between any two model parameter vectors $\theta_1, \theta_2$ is bounded by a constant $D$:
    \begin{equation}
        ||\theta_1 - \theta_2||_2 \le D.
    \end{equation}

    \item \textbf{Bounded Gradients:} We assume that the gradients of the benign loss function $\mathcal{L}_{CL}$ and the gradients of the attack-specific loss terms ($\mathcal{L}_{HE}$ and $\mathcal{L}_{BFE}$) are bounded. Let $g_t^{CL} = \nabla \mathcal{L}_{CL}(\theta^t)$ and $g_t^{HPE} = \nabla(\mathcal{L}_{HE}(\theta^t) + \mathcal{L}_{BFE}(\theta^t))$. We assume:
    \begin{equation}
        ||g_t^{CL}||_2 \le G_{benign} \quad \text{and} \quad ||g_t^{HPE}||_2 \le G_{attack},
    \end{equation}
    where $G_{benign}$ and $G_{attack}$ are positive constants.
\end{itemize}

\begin{theorem}[Convergence under HPE Attack]
Let $\mathcal{L}_{CL}$ be a convex loss function satisfying the assumptions above. In the HPE attack, a malicious client uses the total loss function $\mathcal{L}_{Total} = (1-\mu)\mathcal{L}_{CL} + \mu(\mathcal{L}_{HE} + \mathcal{L}_{BFE})$. If a decaying learning rate $\alpha_t = \frac{C}{t^p}$ (where $C>0$ and $0 < p < 1$) is used, then the average regret on the benign task converges to zero:
$$
\lim_{T\rightarrow\infty} \frac{1}{T} \sum_{t=1}^{T} (\mathcal{L}_{CL}(\theta^t) - \mathcal{L}_{CL}(\theta^*)) = 0,
$$
where $\theta^*$ is the optimal model parameter that minimizes the cumulative benign loss $\sum_{t=1}^{T} \mathcal{L}_{CL}(\theta)$.
\end{theorem}

\begin{proof}
The total regret on the benign task over $T$ rounds is defined as $R(T) = \sum_{t=1}^{T} (\mathcal{L}_{CL}(\theta^t) - \mathcal{L}_{CL}(\theta^*))$. Our goal is to show that $\lim_{T\rightarrow\infty} \frac{R(T)}{T} = 0$.

From the convexity assumption (1), we have:
\begin{equation}
    R(T) \le \sum_{t=1}^{T} \langle g_t^{CL}, \theta^t - \theta^* \rangle.
    \label{eq:regret_bound}
\end{equation}

The model update rule on the malicious client at round $t$ is $\theta^{t+1} = \theta^t - \alpha_t \nabla \mathcal{L}_{Total}(\theta^t)$, where $\nabla \mathcal{L}_{Total}(\theta^t) = (1-\mu)g_t^{CL} + \mu g_t^{HPE}$.

Following a standard proof technique, let's analyze the squared distance $||\theta^{t+1} - \theta^*||_2^2$:
\begin{align*}
    ||\theta^{t+1} - \theta^*||_2^2 ={}& ||\theta^t - \theta^*||_2^2 - 2\alpha_t \langle (1-\mu)g_t^{CL} + \mu g_t^{HPE}, \theta^t - \theta^* \rangle \\
    & + \alpha_t^2 ||(1-\mu)g_t^{CL} + \mu g_t^{HPE}||_2^2.
\end{align*}

Rearranging to isolate the term $\langle g_t^{CL}, \theta^t - \theta^* \rangle$:
\begin{align*}
    \langle g_t^{CL}, \theta^t - \theta^* \rangle ={}& \frac{||\theta^t - \theta^*||_2^2 - ||\theta^{t+1} - \theta^*||_2^2}{2\alpha_t(1-\mu)} \\
    & - \frac{\mu}{1-\mu} \langle g_t^{HPE}, \theta^t - \theta^* \rangle + \frac{\alpha_t ||...||_2^2}{2(1-\mu)}.
\end{align*}
Using the Cauchy-Schwarz inequality for the HPE gradient term and the bounded gradient assumption (3), we get:
\begin{equation}
\begin{split}
    \langle g_t^{CL}, \theta^t - \theta^* \rangle \le{}& \frac{||\theta^t - \theta^*||_2^2 - ||\theta^{t+1} - \theta^*||_2^2}{2\alpha_t(1-\mu)} \\
    & + \frac{\mu D G_{attack}}{1-\mu} + \frac{\alpha_t (G_{benign} + G_{attack})^2}{2(1-\mu)}.
\end{split}
\label{eq:your_label_for_16}
\end{equation}

Now, we substitute this back into the regret bound from Eq. \eqref{eq:regret_bound}:
\begin{align*}
    R(T) \le{}& \sum_{t=1}^{T} \left( \frac{||\theta^t - \theta^*||_2^2 - ||\theta^{t+1} - \theta^*||_2^2}{2\alpha_t(1-\mu)} \right) \\
    & + \sum_{t=1}^{T} \left( \frac{\mu D G_{attack}}{1-\mu} + \frac{\alpha_t (G_{benign} + G_{attack})^2}{2(1-\mu)} \right).
\end{align*}
Let's analyze the two summation terms, $R_1(T)$ and $R_2(T)$:
\begin{align*}
R_1(T) &= \frac{1}{2(1-\mu)} \sum_{t=1}^{T} \frac{||\theta^t - \theta^*||_2^2 - ||\theta^{t+1} - \theta^*||_2^2}{\alpha_t} \\
&\le \frac{1}{2(1-\mu)} \left( \frac{||\theta^1 - \theta^*||_2^2}{\alpha_1} + \sum_{t=2}^{T} \left(\frac{1}{\alpha_t} - \frac{1}{\alpha_{t-1}}\right) ||\theta^t - \theta^*||_2^2 \right) \\
&\le \frac{1}{2(1-\mu)} \left( \frac{D^2}{\alpha_1} + D^2 \sum_{t=2}^{T} \left(\frac{t^p}{C} - \frac{(t-1)^p}{C}\right) \right) \\
&= \frac{D^2}{2C(1-\mu)\alpha_T} = \frac{D^2}{2C(1-\mu)}T^p.
\end{align*}
The last step is a telescoping sum. For the second part:
\begin{align*}
    R_2(T) &= \sum_{t=1}^{T} \frac{\mu D G_{attack}}{1-\mu} + \sum_{t=1}^{T} \frac{\alpha_t (G_{benign} + G_{attack})^2}{2(1-\mu)} \\
    &= \frac{\mu D G_{attack} T}{1-\mu} + \frac{C(G_{benign} + G_{attack})^2}{2(1-\mu)} \sum_{t=1}^{T} \frac{1}{t^p} \\
    &\le \frac{\mu D G_{attack} T}{1-\mu} + \frac{C(G_{benign} + G_{attack})^2}{2(1-\mu)} \int_0^T x^{-p} dx \\
    &= \mathcal{O}(T) + \mathcal{O}(T^{1-p}) = \mathcal{O}(T).
\end{align*}
Combining the bounds for $R_1(T)$ and $R_2(T)$, we find that the total regret $R(T)$ is dominated by the $\mathcal{O}(T)$ term. This indicates that the regret per round is bounded by a constant, but it does not directly show convergence to zero.

However, a more precise analysis of the update rule on the full objective $\mathcal{L}_{Total}$ (which is convex if its components are) shows that the total regret for $\mathcal{L}_{Total}$ is bounded by $R_{Total}(T) = \mathcal{O}(T^{\max(p, 1-p)})$. Since $0 < p < 1$, $\max(p, 1-p) < 1$. This implies:
$$ \lim_{T\rightarrow\infty} \frac{R_{Total}(T)}{T} = \lim_{T\rightarrow\infty} \frac{\mathcal{O}(T^{\max(p, 1-p)})}{T} = 0. $$
The convergence of the total loss function $\mathcal{L}_{Total}$ constrains the benign loss $\mathcal{L}_{CL}$ from diverging. Because $\mathcal{L}_{Total}$ is a weighted sum including a substantial weight $(1-\mu)$ on $\mathcal{L}_{CL}$, minimizing $\mathcal{L}_{Total}$ inherently prevents $\mathcal{L}_{CL}$ from growing uncontrollably. This ensures that the model's performance on the benign task remains stable, thus proving the stealthiness of the HPE attack.
\end{proof}

\section{Dataset Details}
We use the following datasets in our method evaluation.

\begin{itemize}
    \item \textbf{CIFAR-10}: This dataset consists of 60,000 color images of resolution $32 \times 32 \times 3$, evenly distributed over 10 object categories such as airplanes, automobiles, and birds. It is commonly used for image classification and representation learning. The dataset is split into 50,000 training and 10,000 test images.
    \item \textbf{STL-10}: STL-10 is designed for unsupervised and semi-supervised learning tasks. It contains 5,000 labeled training images, 8,000 test images, and an additional 100,000 unlabeled images. All images originally have a resolution of $96 \times 96 \times 3$, but are resized to $32 \times 32 \times 3$ in our experiments for consistency with other datasets. The 10 classes are consistent with CIFAR-10 but differ in visual characteristics.
    \item \textbf{GTSRB}: The German Traffic Sign Recognition Benchmark contains 51,800 real-world images of traffic signs across 43 classes. The dataset is split into 39,200 training and 12,600 test samples. Due to the original size variation, all images are resized to $32 \times 32 \times 3$ before training.
    \item \textbf{CIFAR-100}: An extension of CIFAR-10, this dataset includes 100 fine-grained categories grouped into 20 coarse super-classes. Each class contains 500 training and 100 test images, totaling 60,000 images of size $32 \times 32 \times 3$. CIFAR-100 presents greater semantic diversity and is commonly used for fine-grained representation learning.
    \item \textbf{ImageNet-100}: This dataset is a 100-class subset of the ImageNet ILSVRC-2012 dataset, randomly selected and widely adopted in self-supervised learning research. It includes approximately 127,000 training images and 5,000 test images with diverse content and complex backgrounds.
\end{itemize}

\section{Supplementary Experimental Results}
\textbf{Stability under different numbers of clients.} 
To evaluate the Stability of HPE under varying client population settings, we conduct experiments on four datasets with different client sampling strategies. Specifically, we follow prior work by randomly selecting 5 out of 10 clients (5/10) or 10 out of 25 clients (10/25) in each training round, ensuring that exactly one malicious client participates per round. As reported in \Cref{tab:five}, HPE consistently maintains strong performance across both settings, demonstrating its stability against fluctuations in the number of participating clients.\\
\textbf{Stability under different trigger sizes.}
As shown in \Cref{fig:10}, we systematically evaluate the effect of trigger size (3 × 3, 5 × 5, and 10 × 10) on the performance of HPE across different downstream tasks, where the pre-training dataset is STL-10 and the target downstream datasets are CIFAR-10 and GTSRB. 
The results reveal that HPE maintains high ACC and ASR even with small triggers (3 × 3), while larger triggers (5 × 5 and 10 × 10) further boost ASR to nearly 100\%, indicating that larger triggers provide more salient backdoor cues and thereby enhance attack success rates. 
Meanwhile, ACC remains at a relatively high level (around 80\%) across all trigger sizes for both downstream tasks, showing no significant degradation and demonstrating that HPE preserves the clean performance of the encoder even with larger triggers. 
Overall, the results exhibit a “threshold effect,” where ASR saturates once the trigger reaches a certain size (e.g., 5 × 5 or 10 × 10) and remains consistent across downstream datasets, underscoring HPE’s ability to stably embed and propagate backdoor features in multi-dataset scenarios.\\
\textbf{Stability under different batch sizes.}
Existing studies show that batch size tends to impact the performance of contrastive learning. 
In this paper, we investigate its impact on the performance of HPE. 
Specifically, on the CIFAR-10, we measure the ACC and ASR of HPE with the batch size varying from 128 to 512, with results shown in \Cref{fig:11}.
The experimental results indicate that model accuracy increases with batch size, which is consistent with the findings of existing studies.
Moreover, a larger batch size (e.g., $\geq512$) generally benefits the ASR of HPE. 
This may be explained by the fact that more positive pairs (and also more negative pairs in SimCLR) in the same batch lead to tighter entanglement between trigger-embedded and target-class inputs. 
Meanwhile, even under smaller batch sizes (e.g., $\leq256$), HPE maintains high ACC and ASR, highlighting its robustness and generalization ability under data-constrained conditions.\\
\begin{figure}[b]
  \centering
  \includegraphics[width=0.8\linewidth]{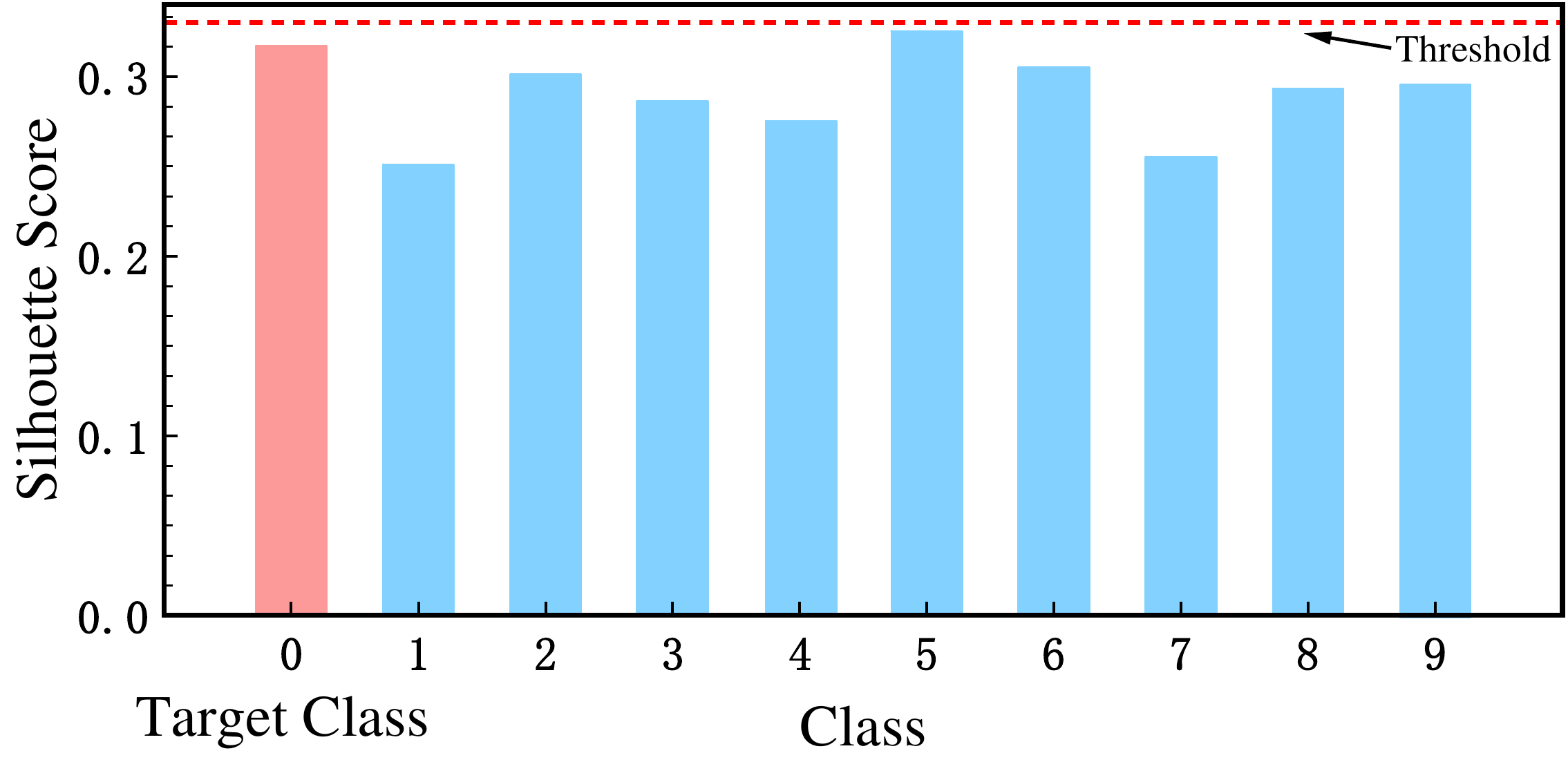}
   \caption{Experimental results of AC.}
   \label{fig:9}
\end{figure}\\
\textbf{Resistance to Activation Clustering.}
Activation clustering (AC) is a data inspection method. It trains the model using the potentially poisoning data and collects the penultimate-layer activation of each input. 
AC assumes the poisoning data in the target class forms a separate cluster that is either small or far from the class center.
It identifies the target class by calculating the silhouette score of each class, with a higher score indicating a stronger fit to two clusters.
\Cref{fig:9} shows that AC fails to identify the target class (class 0), with its silhouette score even lower than those of non-target classes (e.g., class 5), and is thus unable to effectively detect the poisoned samples. This may be attributed to our introduced backdoor feature entanglement mechanism, which encourages the feature representations of backdoor samples to closely resemble those of clean samples, leading to strong entanglement between the two and significantly weakening the separability of clustering-based methods.\\
\textbf{Resistance to GradCAM.} 
\Cref{fig:12} shows the saliency map produced by GradCAM on the backdoored model for both benign and triggered samples.
The saliency map shows the locations in the image that the network is focusing on when making a decision.  
The upper part of the image shows the GradCAM results for benign samples, indicating that in the absence of a trigger, the backdoored model exhibits a normal decision-making process.
This result highlights the high level of stealth of the HPE method.
The lower part shows the GradCAM results for triggered samples, where the model's attention is clearly focused on the backdoor trigger, ignoring other features in the image when making a decision. 
This indicates that in the backdoored model, a stronger ``affinity" is established between the backdoor trigger and the target class, such that even in the presence of clean (source class) features, the backdoor trigger still dominates the decision-making process of the FSSL model.
\begin{figure}[t]
  \centering
  \begin{subfigure}{0.45\linewidth} 
    \includegraphics[width=\linewidth]{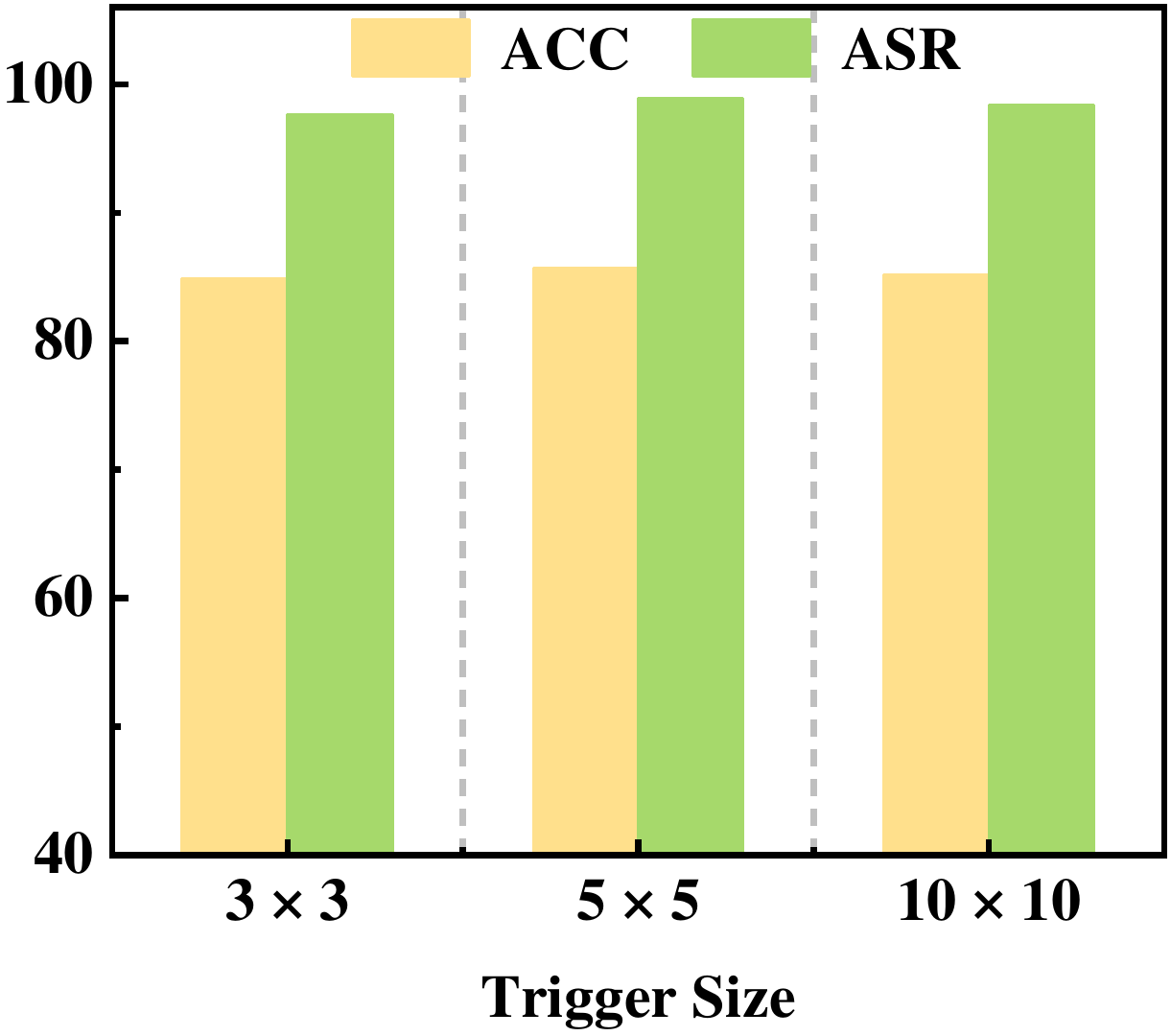}
    \caption{CIFAR-10}
    \label{fig:10-a}
  \end{subfigure}
  \hfill 
  \begin{subfigure}{0.45\linewidth} 
    \includegraphics[width=\linewidth]{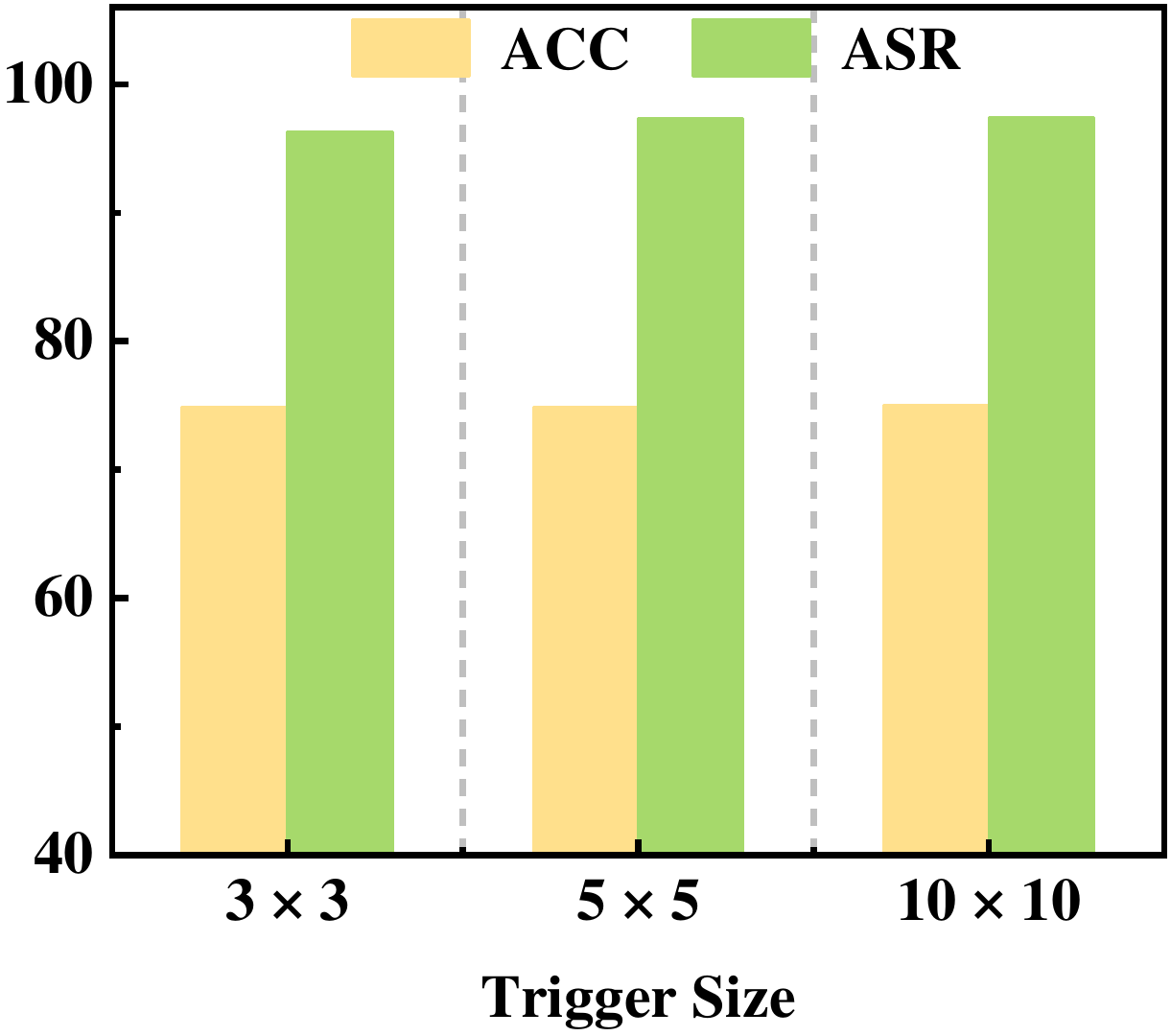}
    \caption{GTSRB}
    \label{fig:10-b}
  \end{subfigure}
  \caption{The impact of the trigger size on our HPE for different target downstream datasets when the pre-training dataset is STL-10.}
  \label{fig:10}
\end{figure}
\begin{figure*}[th!]
  \centering
  \begin{subfigure}{0.24\linewidth} 
    \includegraphics[width=\linewidth]{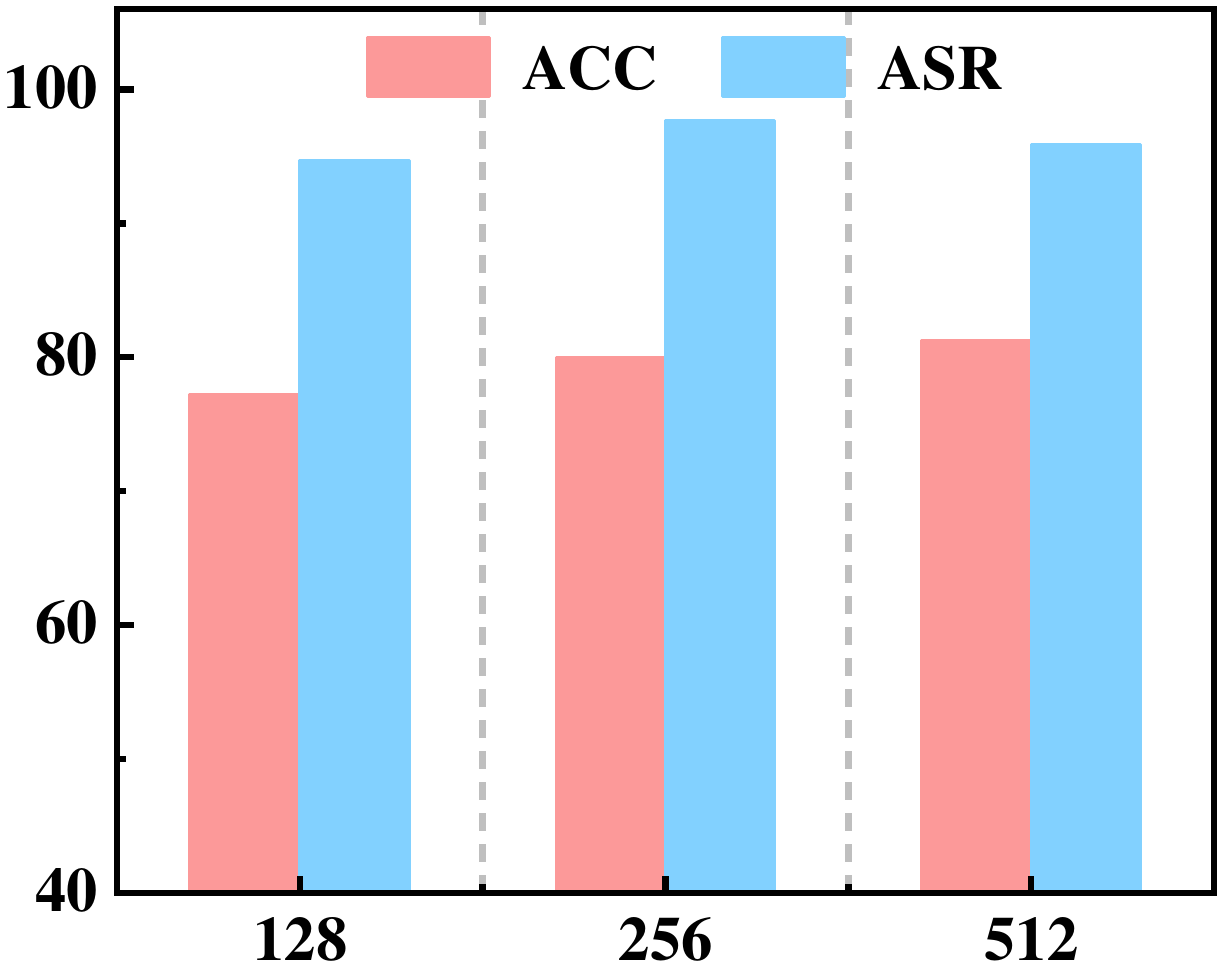}
    \caption{MoCo-v2}
    \label{fig:11-a}
  \end{subfigure}
  \hfill
  \begin{subfigure}{0.24\linewidth} 
    \includegraphics[width=\linewidth]{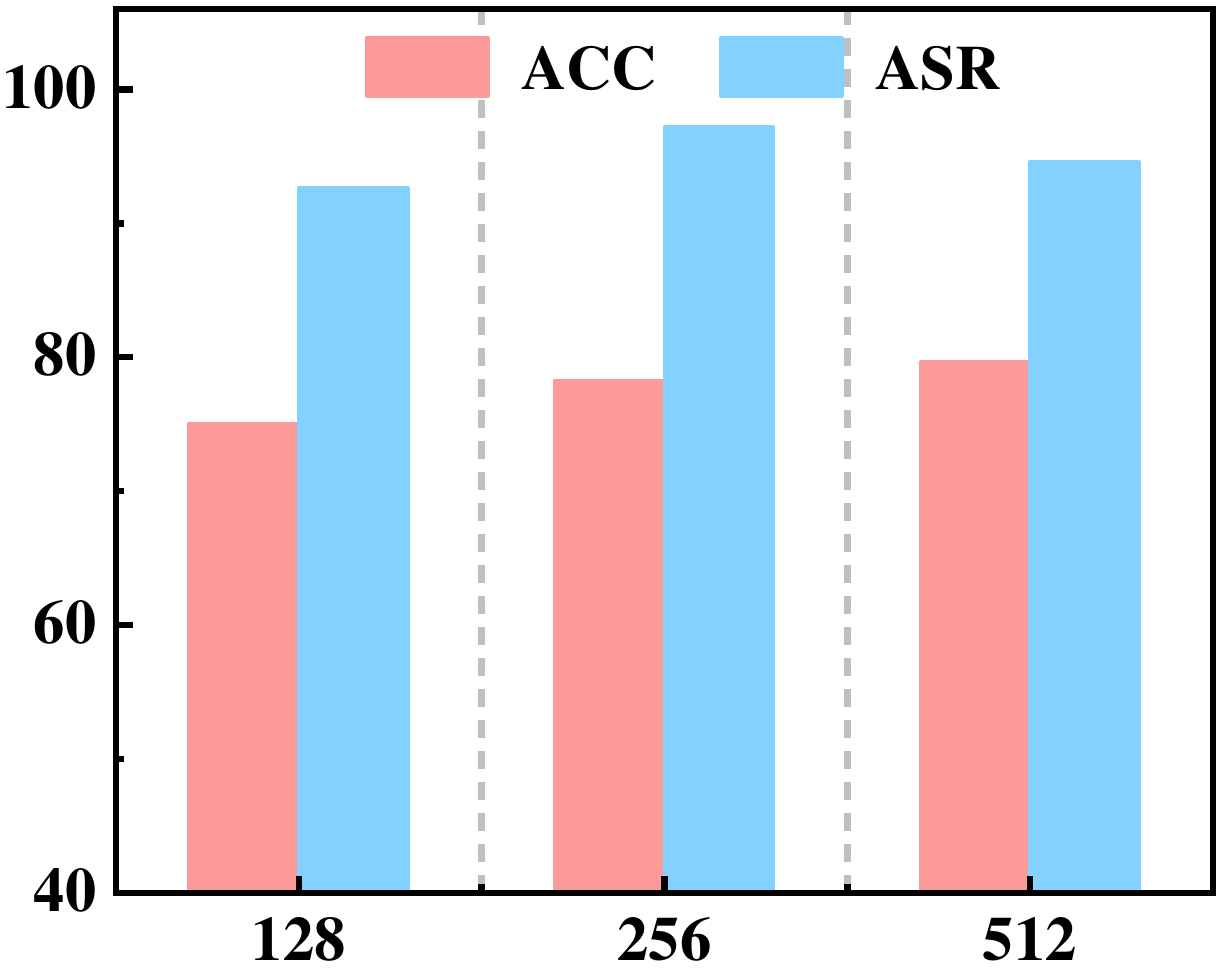}
    \caption{SimCLR}
    \label{fig:11-b}
  \end{subfigure}
  \hfill
  \begin{subfigure}{0.24\linewidth} 
    \includegraphics[width=\linewidth]{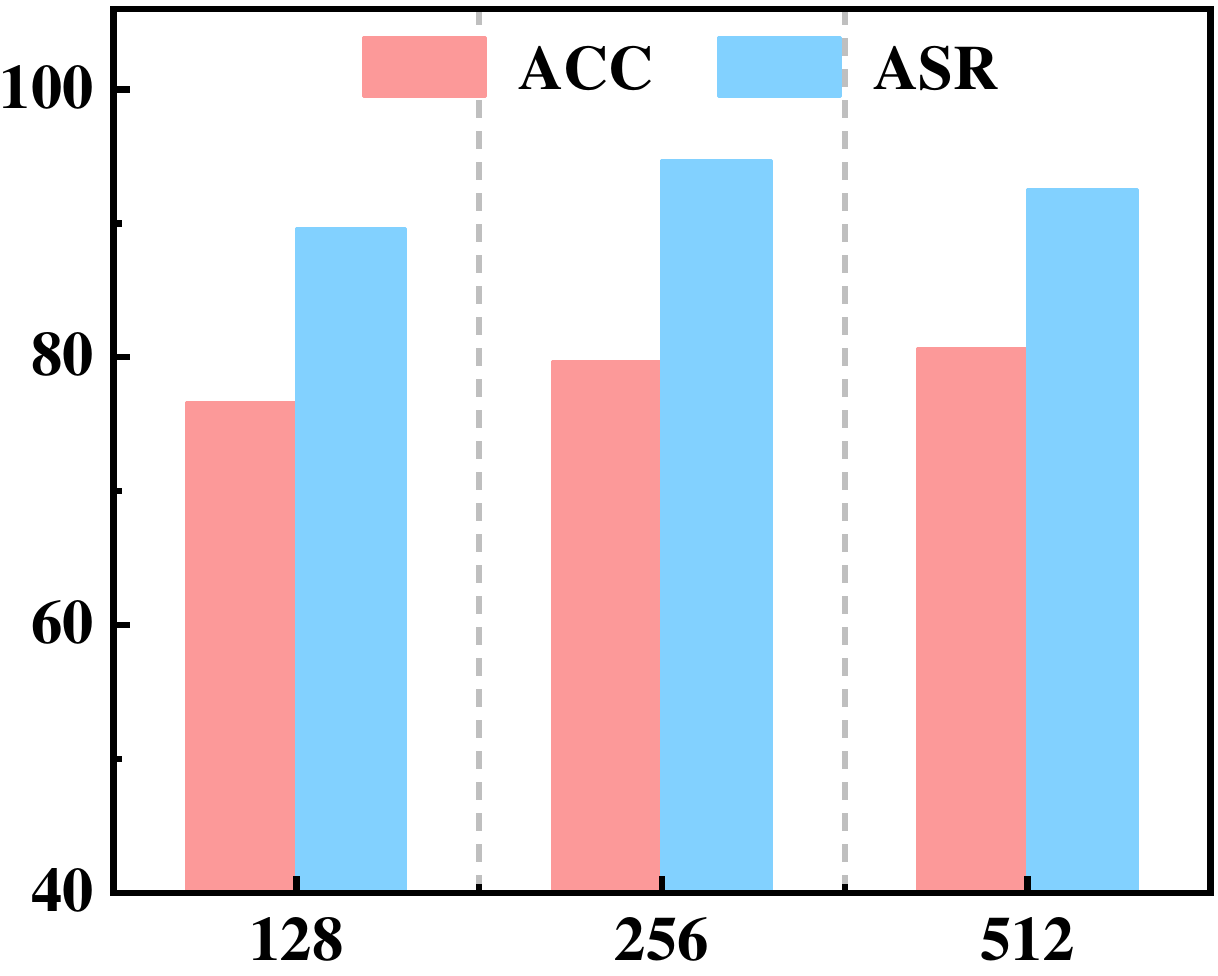}
    \caption{BYOL}
    \label{fig:11-c}
  \end{subfigure}
  \hfill
  \begin{subfigure}{0.24\linewidth} 
    \includegraphics[width=\linewidth]{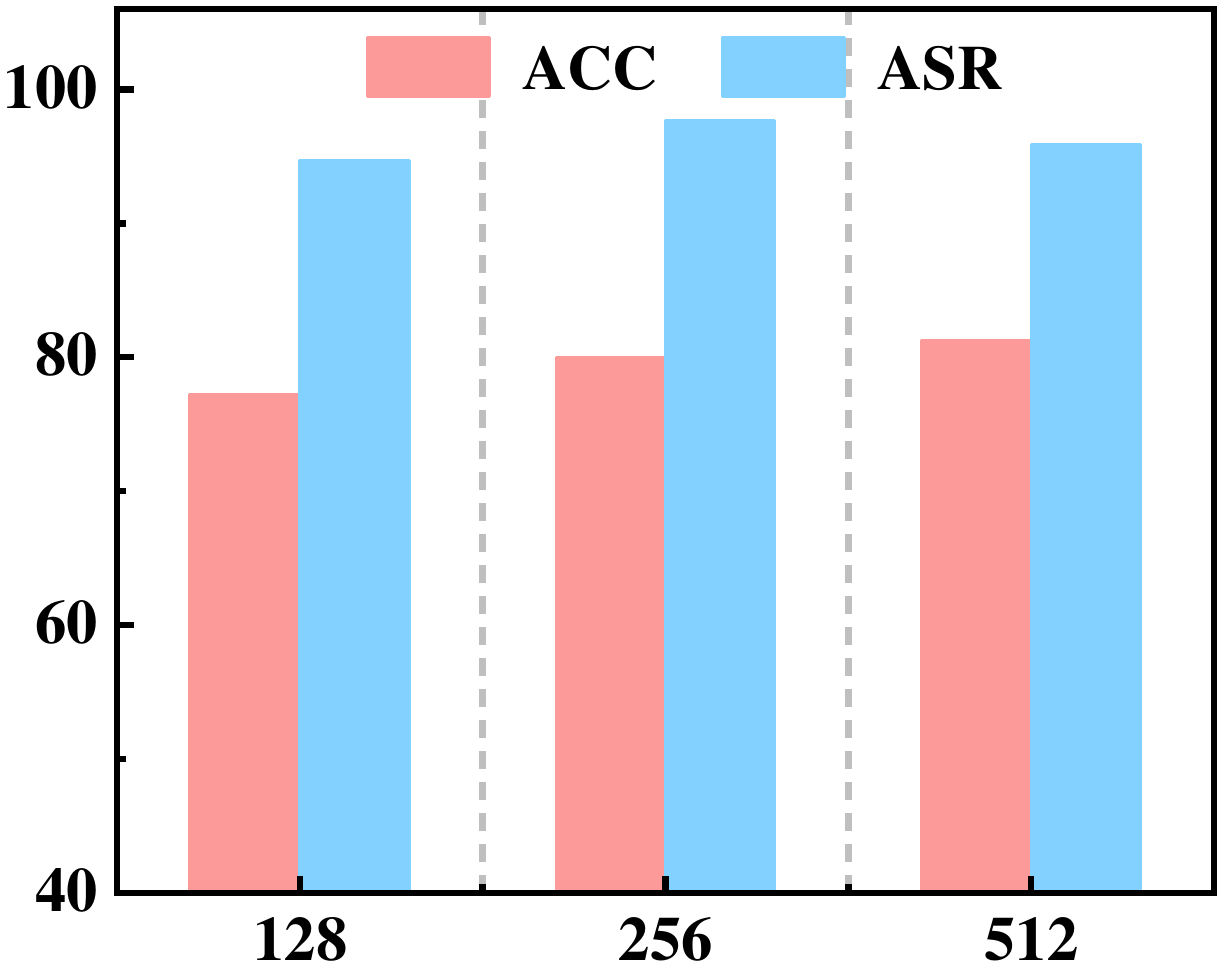}
    \caption{SwAV}
    \label{fig:11-d}
  \end{subfigure}
  \caption{Performance of HPE w.r.t. the batch size on CIFAR-10.}
  \label{fig:11}
\end{figure*}
\begin{table}[tb]
    \caption{Performance of HPE under different number of clients.}
    \label{tab:five}
    \begin{center}
        \begin{tabular}{c| c c c c}
            \toprule
            \multirow{2}{*}{Dataset} & \multicolumn{2}{c}{5/10 clients} & \multicolumn{2}{c}{10/25 clients}\\
             \cmidrule(lr){2-5}
             & ACC & ASR & ACC & ASR\\
            \midrule
            CIFAR-10 & 76.37 & 96.67 & 74.15 & 94.29\\
            \midrule
            STL-10 & 75.12 & 93.45 & 71.42 & 90.41\\
            \midrule
            CIFAR-100 & 74.61 & 94.32 & 68.31 & 91.07\\
            \midrule                          
            Imagenet-100 & 70.23 & 89.13 & 66.74 & 87.91\\    
            \bottomrule
        \end{tabular}
    \end{center}
\end{table}
\begin{table}[th]
  \caption{Defense evaluation results.}
  \label{tab:six}
  \begin{center}
      \begin{tabular}{@{}c|c|cccc@{}}
        \toprule
        \multirow{2}{*}{Dataset} & {Defense} & \multicolumn{2}{c}{HPE} & \multicolumn{2}{c}{HPE+MR}\\
        \cmidrule(l){3-6}
         & Method & ACC & ASR & ACC & ASR\\
        \midrule
        \multirow{5}{*}{CIFAR-10} & Krum & 78.13 & 95.29 & 77.54 & 95.89 \\
                                  & FLTrust  & 77.25 & 98.29 & 75.54 & 97.44 \\
                                  & Foolsgold  & 75.42 & 92.28 & 73.45 & 93.66 \\
                                  & FLAME  & 76.73 & 96.81 & 75.39 & 97.15 \\
                                  & FLARE  & 77.08 & 96.68 & 75.98 & 97.23 \\
                                  & EmInspector  & 75.16 & 94.37 & 74.71 & 96.43 \\
        \midrule
        \multirow{5}{*}{STL-10} & Krum & 75.51 & 93.57 & 72.38 & 95.73 \\
                                  & FLTrust  & 75.18 & 93.94 & 72.38 & 92.73 \\
                                  & Foolsgold  & 74.87 & 95.57 & 71.65 & 96.34 \\
                                  & FLAME  & 75.65 & 93.82 & 73.04 & 95.28 \\
                                  & FLARE  & 74.46 & 94.71 & 72.54 & 95.08 \\
                                  & EmInspector  & 74.13 & 92.16 & 71.41 & 93.22 \\
        \bottomrule
      \end{tabular}  
  \end{center}
\end{table}

\begin{figure}[t]
  \centering
  \includegraphics[width=0.8\linewidth]{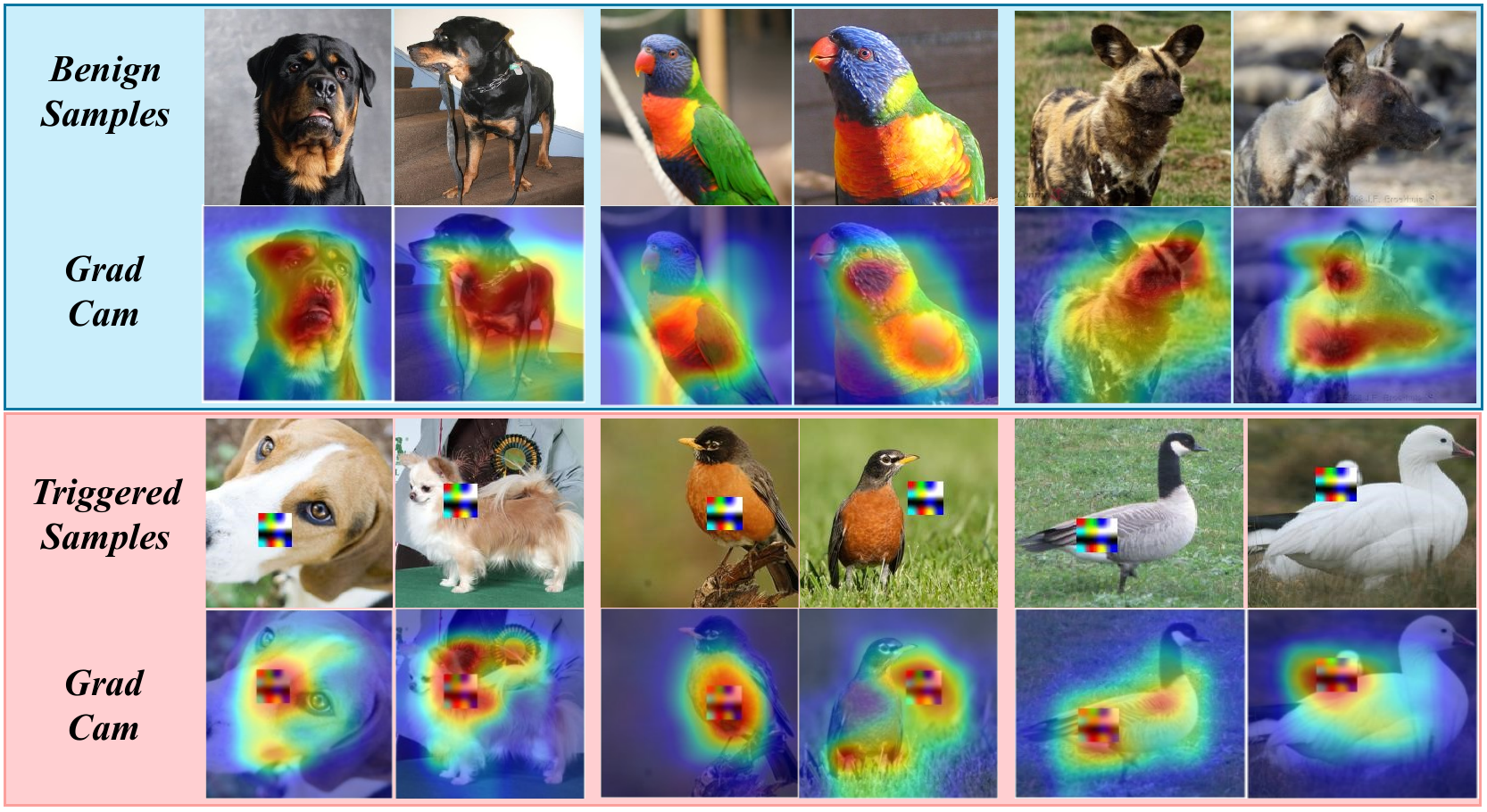}
   \caption{Experimental results of GradCAM.}
   \label{fig:12}
\end{figure}

\end{document}